\documentclass[11pt, a4paper]{article}
\usepackage{amsmath, amssymb, amsthm, amsfonts, mathrsfs}
\usepackage{geometry}
\usepackage{float}
\usepackage{hyperref}
\usepackage{authblk}
\usepackage{graphicx}
\usepackage{array}
\usepackage{titlesec}
\usepackage{needspace}
\usepackage[utf8]{inputenc}
\usepackage[T1]{fontenc}
\usepackage{newtxtext, newtxmath} 
\usepackage{enumitem}
\usepackage{braket}
\usepackage{tikz}
\usetikzlibrary{cd, positioning, arrows.meta}
\usepackage{mathtools}
\usepackage{xcolor}
\usepackage{comment}
\usepackage[backend=biber,style=numeric,sorting=none,sortcites=true,giveninits=true,maxbibnames=10,minbibnames=10,doi=true,isbn=false,url=false]{biblatex}
\DeclareFieldFormat[article]{volume}{\mkbibbold{#1}}

\hypersetup{
    colorlinks=true,
    linkcolor=blue,
    filecolor=magenta,      
    urlcolor=cyan,
    citecolor=red,
}

\theoremstyle{definition}
\newtheorem{theorem}{Theorem}[section]
\newtheorem{lemma}[theorem]{Lemma}
\newtheorem{proposition}[theorem]{Proposition}
\newtheorem{corollary}[theorem]{Corollary}
\newtheorem{definition}[theorem]{Definition}
\newtheorem{example}[theorem]{Example}
\newtheorem{remark}[theorem]{Remark}

\newcommand{\mf}[1]{\mathfrak{#1}}

\DeclareMathOperator{\Tr}{Tr}
\DeclareMathOperator{\ad}{ad}
\DeclareMathOperator{\Ad}{Ad}

\DeclareMathOperator{\Span}{span}
\DeclareMathOperator{\Ker}{Ker}

\DeclareMathOperator{\cone}{cone}
\DeclareMathOperator{\aff}{\mathfrak{aff}}
\DeclareMathOperator{\Lie}{Lie}

\newcommand{\Lcal}{\mathcal{L}}
\newcommand{\Hcal}{\mathcal{H}}
\newcommand{\Bcal}{\mathcal{B}}

\newcommand{\Ecal}{\mathcal{E}}

\newcommand{\Dcal}{\mathcal{D}}

\newcommand{\g}{\mathfrak{g}}
\newcommand{\s}{\mathfrak{s}}
\newcommand{\rr}{\mathfrak{r}}

\newcommand{\gl}{\mathfrak{gl}}
\newcommand{\sll}{\mathfrak{sl}}

\newcommand{\su}{\mathfrak{su}}
\newcommand{\so}{\mathfrak{so}}
\DeclareMathOperator{\rad}{rad}

\newcommand{\C}{\mathbb{C}}
\newcommand{\R}{\mathbb{R}}

\newcommand{\WGKLS}{\mathcal{W}_{\text{GKLS}}}

\newcommand{\Wrad}{W_{\rr}}
\newcommand{\Cquot}{C_{\text{quot}}}

\newcommand{\levi}{\ltimes} 
\newcommand{\gLK}{\mf{g}_{LK}}
\newcommand{\Htr}{\Hcal_{0}} 

\begin{document}

\title{\textbf{Lie-Wedge Stratification and Pure-State Stabilizability in Single-Channel Qubit Control}}

\author{Yibin Wang}
\affil[1]{Graduate School of Mathematics, Nagoya University, Nagoya, 464-8601, Aichi, Japan.}
\affil[ ]{E-mail: yibinw0210@gmail.com}

\date{}

\maketitle

\begin{abstract}
We determine when qubit controls admit an effective generator that drives every initial state to the same pure state. Lie wedges enforce nonnegative weights for dissipative generators and thereby preserve the distinction between reversible control and irreversible dissipation lost under Lie-algebra closure. For drift-free qubits with Hamiltonian controls available with either sign and at most one independently actuated dissipator generated by a traceless jump operator, we classify all nontrivial systems into thirteen structural strata. Explicit tests using the control algebra and channel data assign each system to a stratum. Necessary and sufficient conditions determine whether its Lie wedge contains a stabilizing generator and, when it does, the minimum number of control-rotated dissipative orientations required. These conditions distinguish stabilizable and non-stabilizable systems even when they share the same generated Lie algebra and stratum. In an exceptional stratum, one dissipative orientation is insufficient even with the available Hamiltonian controls, whereas averaging two orientations of the same channel produces a stabilizing generator.
\end{abstract}

\section{Introduction}
\label{sec:introduction}

We study when a controlled dissipative qubit admits an effective generator with a unique globally attractive pure state. In closed quantum control, the dynamical Lie algebra records the directions generated by the available Hamiltonians and their commutators. For the bidirectional Hamiltonian controls considered here, this closure retains the reversible coherent directions. Markovian dissipation changes the geometry: a valid Gorini--Kossakowski--Lindblad--Sudarshan (GKLS) generator need not remain valid after its sign is reversed. Passing directly to the generated Lie algebra therefore erases the distinction between reversible and irreversible directions. Two control specifications may generate the same Lie algebra while supporting different dissipative directions and different pure-state stabilization capabilities.

To retain the one-sided constraint, we attach two objects to a specification \(G\) that records independently actuable generators and their allowed signs. The system algebra \(\g(G)\) records bracket generation. The minimal closed Lie wedge \(W(G)\) is the smallest closed convex cone in \(\g(G)\) containing the nonnegative multiples of the declared generators and invariant under conjugation by its reversible part, the edge. A Levi decomposition separates the solvable radical from a semisimple factor and organizes the wedge data across these components. The radical is the largest ideal whose iterated commutator algebras eventually vanish. The semisimple factor has no nonzero solvable ideals. This poses a structural question that the Lie algebra alone cannot answer: which Levi data can occur for a given physical family, and how much control information do those data retain?

Lie-semigroup and Lie-wedge methods provide the geometric framework for controlled Markovian dynamics \cite{dirr2009lie,o2025lie}. O'Meara, Dirr, and Schulte-Herbr\"uggen computed explicit global wedges for standard coherently controlled unital qubit channels \cite{OMeara2011IllustratingTG}. Complementary work uses fast unitary reduction to study time-optimal cooling, heating, and stabilizable states of a Markovian qubit \cite{malvetti2024optimal}. The Levi-type decomposition and reconstruction used here are established in Ref.~\cite[Sec.~4]{wangthakkar2026geometry}. We apply those tools to the generator-labelled wedge of the full state dynamics, retaining the independent actuation of coherent and dissipative rays.

Pure-state stabilization supplies a control task for the classification. O'Meara's purely dissipative criteria distinguish pure fixed points from unique fixed points \cite[Cor.~3.2.2 and Prop.~3.2.2]{o2025lie}. Ticozzi, Lucchese, Cappellaro, and Viola treat time-independent Hamiltonian stabilization of fixed Lindblad noise under general and restricted control resources \cite{ticozzi2012hamiltonian}. Under fast full unitary control, Malvetti, vom Ende, Dirr, and Schulte-Herbr\"uggen show that a unique attractive pure fixed point exists exactly when, for each choice of Lindblad terms, the terms have a common right eigenvector that is not a common left eigenvector \cite[Thm.~4.7]{malvetti2024reachability}. Our restricted single-channel qubit analysis gives a fixed-axis scalar criterion, locates positive and negative cases within the thirteen strata, and determines the minimum number of control-rotated dissipative directions required.

\Needspace{5\baselineskip}
\paragraph{Main results.}
The classification gives a finite map of the possible control structures. For the drift-free qubit family with bidirectional Hamiltonian controls and at most one independently actuated dissipator ray from a traceless jump, the coherent-control algebra and four channel discriminants assign every nontrivial specification to one of thirteen Levi-data strata, realizing all four Levi types (Theorem~\ref{thm:qubit_classification}). In every stratum, the reversible directions are exactly the coherent-control algebra. Each classified wedge is also the tangent wedge of its generated semigroup, connecting the infinitesimal geometry to the closed reachable dynamics. Corollary~\ref{cor:qubit_globality_semialgebra} identifies four Lie-semialgebra strata: locally, composing flows from these wedges still gives a logarithmic generator in the same wedge.

The stabilization atlas turns this structural map into a state-preparation test (Corollary~\ref{cor:qubit_pure_gas_atlas}). It asks whether the wedge contains a generator that drives every initial state to the same pure state. Five strata contain both stabilizable and non-stabilizable specifications, distinguished by additional channel data. For the interior fixed-axis branches, Theorem~\ref{thm:axial_pure_gas} resolves this dependence through a single scalar condition and determines the minimum number of control-rotated dissipative directions needed.

Changing the relative control axis can change pure-state stabilizability even when the generated algebra and stratum remain the same (Example~\ref{ex:q9_control_contrast}). Averaging two dissipative orientations in the wedge can also yield a pure-state attractor when no single orientation suffices (Example~\ref{ex:q6_half_turn}).

Section~\ref{sec:control_consequences} introduces the qubit model through these examples. Section~\ref{sec:structural_recap} defines the structural data, and Section~\ref{sec:qubit_classification} gives the thirteen-stratum classification, its orbit geometry, and globality. Section~\ref{sec:pure_stabilization} then derives the stabilization criterion and atlas. Section~\ref{sec:physical_consequences} explains what the strata determine and which questions require finer data. The appendices retain the general wedge construction, normalization and structural proofs, Levi reconstruction, and parameter-by-parameter classification calculations.

\Needspace{8\baselineskip}
\section{Qubit Model and Motivating Control Examples}
\label{sec:control_consequences}

Can two control specifications with the same generated Lie algebra differ in pure-state stabilizability? Can averaging dissipative orientations stabilize a pure state that no single orientation can stabilize? The examples below answer these questions for a drift-free qubit with bidirectional Hamiltonian control about a unit axis $n$ and an independently actuated dissipator $\Dcal_L$ from one traceless jump. Each declared ray can be switched off and scaled freely. The Hamiltonian amplitude ranges over all real values, and no additional drift generator is imposed.

Let $\sigma=(\sigma_x,\sigma_y,\sigma_z)$ be the Pauli vector and write the state as $\rho=(I+r\cdot\sigma)/2$. The Bloch ball is $|r|\leq1$, with pure states on $|r|=1$. For a jump $L$, the dissipator is
\[
\Dcal_L(\rho)=L\rho L^\dagger-\tfrac12\{L^\dagger L,\rho\},
\]
where $\{A,B\}=AB+BA$. Its Bloch equation is $\dot r=Qr+b$. We identify a generator with this affine pair $(Q,b)$, so a coherent rotation $J_n$ denotes $(J_n,0)$. In a right-handed orthonormal principal frame $(\hat\alpha,\hat\beta,\hat\nu)$, positive channel scaling lets us write
\[
L=\sigma_{\hat\alpha}+iq\,\sigma_{\hat\beta},
\qquad 0\leq q\leq1.
\]
Here $\sigma_{\hat\alpha}=\hat\alpha\cdot\sigma$, and similarly for $\hat\beta$. At $q=0$, choose any orthonormal completion of $\hat\alpha$. Corollary~\ref{cor:single_jump_principal_form} establishes this principal form. The Hamiltonian $H=(h/2)n\cdot\sigma$ gives the Bloch rotation $hJ_n$, where $J_n r=n\times r$. The dissipator has Bloch data
\[
Q=\operatorname{diag}\bigl(-2q^2,-2,-2(1+q^2)\bigr),
\qquad b=4q\,\hat\nu.
\]
For a control rotation $R$ about $n$, the conjugate dissipator $\operatorname{Ad}_R\Dcal_L$ has Bloch data $(RQR^T,Rb)$. Its nonnegative multiples form one dissipative orbit ray.
Write $x=|n\cdot\hat\alpha|^2$, $y=|n\cdot\hat\beta|^2$, and $\zeta=|n\cdot\hat\nu|^2$. For $0<q<1$, Theorem~\ref{thm:axial_pure_gas} shows that $W(G)$ contains a generator with pure-state global asymptotic stability (GAS), meaning convergence of every initial density operator to the same pure stationary state in trace norm, exactly when
\begin{equation}
\label{eq:control_task_discriminant}
\Delta:=q^2x-(1-q^2)\zeta=0.
\end{equation}

\begin{example}[The same Lie algebra with different stabilization capabilities]
\label{ex:q9_control_contrast}
Fix $q=1/2$ and choose the two control axes
\[
n_+=\left(\sqrt{\frac3{10}},\sqrt{\frac35},\sqrt{\frac1{10}}\right),
\qquad
n_-=\left(\frac1{\sqrt2},0,\frac1{\sqrt2}\right).
\]
Both specifications generate the same algebra $\gl(3,\R)\ltimes\R^3$ and share the same channel. Table~\ref{tab:qubit_strata} assigns both to stratum $\mathsf{Q9}$. Their relative control axes give the following data.
\begin{center}
\begin{tabular}{c c c c}
\hline
axis & $(x,y,\zeta)$ & $\Delta$ & pure-state GAS in $W(G)$\\
\hline
$n_+$ & $(3/10,3/5,1/10)$ & $0$ & one orbit ray suffices\\
$n_-$ & $(1/2,0,1/2)$ & $-1/4$ & absent\\
\hline
\end{tabular}
\end{center}
For the positive case, the generator $\Dcal_L+(\sqrt5/2)J_{n_+}$ fixes the pure Bloch vector $r_+=(\sqrt3/2,0,1/2)$ and attracts every initial state. For the negative case, the scalar criterion excludes every pure-state GAS generator in the wedge. The common Lie algebra and Levi-data stratum therefore leave the task-relevant control-axis orientation unresolved.
\end{example}

\begin{example}[A pure-state attractor from two orientations]
\label{ex:q6_half_turn}
Take $q=1/\sqrt2$ and $n=(\hat\alpha+\hat\nu)/\sqrt2$. Table~\ref{tab:qubit_strata} assigns this specification to stratum $\mathsf{Q6}$. In this frame,
\[
Q=\operatorname{diag}(-1,-2,-3),\qquad b=2\sqrt2\,\hat\nu.
\]
At the pure vector $r_\ast=n$, the dissipative force is $f_1=Qn+b=(-\hat\alpha+\hat\nu)/\sqrt2$. The half-turn $R_\pi=2nn^T-I$ fixes $n$ and sends this force to $f_2=-f_1$. Thus
\begin{equation}
\label{eq:q6_explicit_average}
X_\ast=\frac12\left(\Dcal_L+\Ad_{R_\pi}\Dcal_L\right),
\qquad
\bar Q=-2I,\qquad \bar b=2n.
\end{equation}
Its trajectories satisfy
\[
\dot r=-2(r-n),\qquad
r(\tau)=n+e^{-2\tau}\bigl(r(0)-n\bigr).
\]
The average has a unique globally attractive pure state. A single orbit ray cannot stabilize it: the axial Hamiltonian force vanishes at $n$, while each individual dissipative force is nonzero. Theorem~\ref{thm:axial_pure_gas} proves that no other pure target admits a one-ray stabilizer in this stratum.
\end{example}

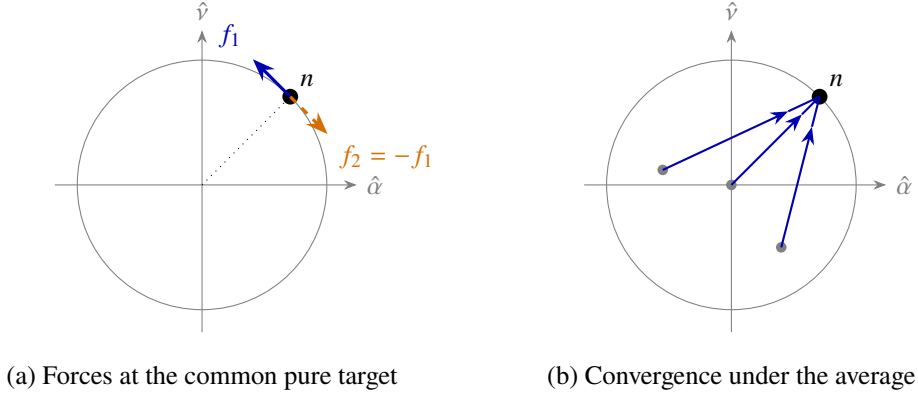
\begin{figure}[H]
\centering
\begin{tikzpicture}[>=Stealth, font=\small]
\begin{scope}[scale=1.65]
\draw[gray] (0,0) circle (1);
\draw[->,gray] (-1.18,0)--(1.24,0) node[right] {$\hat\alpha$};
\draw[->,gray] (0,-1.18)--(0,1.24) node[above] {$\hat\nu$};
\coordinate (targetA) at (0.70710678,0.70710678);
\draw[dotted] (0,0)--(targetA);
\fill (targetA) circle (1.8pt) node[above right] {$n$};
\draw[->,very thick,blue!70!black] (targetA)--(0.40710678,1.00710678)
  node[above left] {$f_1$};
\draw[->,very thick,orange!85!black,dashed] (targetA)--(1.00710678,0.40710678)
  node[below right] {$f_2=-f_1$};
\node at (0,-1.55) {(a) Forces at the common pure target};
\end{scope}
\begin{scope}[xshift=7.0cm,scale=1.65]
\draw[gray] (0,0) circle (1);
\draw[->,gray] (-1.18,0)--(1.24,0) node[right] {$\hat\alpha$};
\draw[->,gray] (0,-1.18)--(0,1.24) node[above] {$\hat\nu$};
\coordinate (targetB) at (0.70710678,0.70710678);
\fill (targetB) circle (1.8pt) node[above right] {$n$};
\foreach \px/\py in {0/0,-0.55/0.12,0.4/-0.5} {
  \fill[gray] (\px,\py) circle (1.2pt);
  \coordinate (nearB) at ({0.2*\px+0.8*0.70710678},{0.2*\py+0.8*0.70710678});
  \draw[blue!70!black,thick] (nearB)--(targetB);
  \draw[->,blue!70!black,thick] (\px,\py)--(nearB);
}
\node at (0,-1.55) {(b) Convergence under the average};
\end{scope}
\end{tikzpicture}
\caption{The $\hat\alpha$--$\hat\nu$ section of the Bloch ball for $\mathsf{Q6}$. The two control-conjugate dissipators exert opposite tangential forces at $n$. Their equal-weight average gives the effective flow $\dot r=-2(r-n)$. Arrows in (b) show the trajectory direction. Convergence to $n$ is asymptotic.}
\label{fig:q6_force_average}
\end{figure}

\begin{remark}[Physical realization of orbit averages]
\label{rem:pure_gas_implementation}
An orbit-ray statement concerns a control-conjugate of the dissipator inside
the Lie wedge. A positive sum of several orbit rays is an effective generator
in the fast-switching or Trotter closure. Corollary~\ref{cor:qubit_globality_semialgebra}
places both constructions in the tangent wedge of the closed reachable
semigroup. Direct implementation as a time-independent laboratory Liouvillian
additionally requires the effective dissipative directions to be independently
addressable.
\end{remark}

\Needspace{8\baselineskip}
\section{Structural Framework}
\label{sec:structural_recap}

The examples distinguish stabilizability within a common Lie algebra and the number of dissipative orientations needed for stabilization. We now define the Lie-wedge and Levi data used to classify these control specifications. The construction of the minimal wedge and the general dissipation identities are deferred to Appendix~\ref{app:wedge_construction}. The definitions below suffice to state the qubit classification.

\subsection{Lie-Wedge Conventions}

We fix the convex-geometric notation used below. A cone $W$ in a real vector space $V$ is closed under multiplication by non-negative scalars. It is convex when it is also closed under addition. A wedge is a closed convex cone. Its edge $E(W)=W\cap(-W)$ is its largest vector subspace. The wedge is pointed when $E(W)=\{0\}$ and generating when $W-W=V$.
For matrix generators, the Lie bracket is $[X,Y]=XY-YX$. The infinitesimal adjoint action is $\ad(X)Y=[X,Y]$, while group conjugation is $\Ad_gX=gXg^{-1}$. These actions satisfy $\Ad_{e^X}=e^{\ad X}$.

\begin{definition}[Lie Wedge]
A closed convex cone $W$ in a Lie algebra $\g$ is a Lie wedge if it satisfies the invariance condition:
\begin{equation}
\label{eq:lie_wedge_condition}
e^{\ad(X)}(W) \subseteq W \quad \text{for all } X \in E(W).
\end{equation}
\end{definition}

\begin{definition}[Lie semialgebra]
\label{def:lie_semialgebra}
A Lie wedge $W\subseteq\g$ is a Lie semialgebra if there is an open neighborhood $B$ of $0$ on which the Baker--Campbell--Hausdorff product $X\star Y=\log(e^Xe^Y)$ is defined and
\[
(W\cap B)\star(W\cap B)\subseteq W.
\]
\end{definition}
This is local closure under logarithms of composed flows, stronger than edge invariance alone \cite[Sec.~1.1, Eq.~(1.5)]{o2025lie}.

\subsection{Ambient Affine Algebra}
Fix a finite-dimensional Hilbert space $\Hcal$ of dimension $d$. Write $\Bcal(\Hcal)$ for its linear operators and $\Htr$ for the real vector space of traceless Hermitian operators. A Hermiticity-preserving map sends Hermitian operators to Hermitian operators. A trace-annihilating generator satisfies $\Tr\Lcal(A)=0$ for every $A$. Its exponential is trace preserving.
\begin{definition}[Lindblad Kossakowski Lie Algebra $\g_{LK}$]
The Lie algebra $\g_{LK}$ is the real Lie algebra consisting of all Hermiticity-preserving and trace-annihilating superoperators on $\Bcal(\Hcal)$.
\end{definition}
\begin{theorem}[Structure of $\g_{LK}$ (affine isomorphism due to O'Meara \cite{o2025lie})]
\label{thm:structure_glk}
The Lindblad Kossakowski Lie algebra $\g_{LK}$ is isomorphic to the Lie algebra of the affine group acting on the space of traceless Hermitian operators,
\begin{equation}
\g_{LK} \cong \aff(\Htr) = \gl(\Htr) \ltimes \Htr,
\end{equation}
where the translational part $\Htr$ is an abelian ideal \cite{o2025lie}.
\end{theorem}
The affine picture is the generalized Bloch representation. Writing a trace-one state as $\rho=I/d+x$, with $x\in\Htr$, gives
\[
\dot x=A_\Lcal x+\frac{1}{d}b_\Lcal,
\qquad b_\Lcal=\Lcal(I).
\]
Thus $b_\Lcal$ is the non-unital translation term and is nonzero exactly when the maximally mixed state is not stationary under the generator.
\begin{remark}[Non-unitality]
Non-unital drift alone does not imply cooling or purification. Those properties also depend on the state, dynamics, and, for cooling, the energy ordering.
\end{remark}
\subsection{The GKLS Wedge}
\label{sec:structure_gkls}

A quantum dynamical semigroup (QDS) is a strongly continuous one-parameter semigroup of completely positive, trace-preserving maps on $\Bcal(\Hcal)$. Complete positivity requires positivity also after tensoring with the identity map on any finite-dimensional ancillary system. The GKLS theorem characterizes its generators \cite{gorini1976completely, lindblad1976generators}. The Kossakowski form separates the Hamiltonian part from the positive dissipative coefficients.

\begin{theorem}[GKLS Theorem (Kossakowski Matrix Form)]
A linear map $\Lcal \in \g_{LK}$ is the generator of a QDS if and only if, in a Hilbert--Schmidt orthonormal basis $\{F_j\}_{j=1}^{d^2-1}$ for $\Htr$ with $\Tr(F_i^\dagger F_j)=\delta_{ij}$, it has the form
\begin{equation}
\label{eq:gkls_kossakowski}
\Lcal(\rho) = -i[H, \rho] + \sum_{j,k=1}^{d^2-1} K_{jk} \left( F_j \rho F_k^\dagger - \frac{1}{2} \{F_k^\dagger F_j, \rho\} \right),
\end{equation}
where $H$ is a Hermitian operator (Hamiltonian), and the Kossakowski matrix $K$ is positive semidefinite ($K \ge 0$).
\end{theorem}
We denote the set of all such generators by $\WGKLS\subseteq\g_{LK}$.
\begin{lemma}[Edge of the GKLS wedge (cf. \cite{dirr2009lie})]
\label{lem:gkls_edge}
The edge of $\WGKLS$ consists precisely of the Hamiltonian generators,
\begin{equation}
E(\WGKLS) = \{\Lcal_H \mid \Lcal_H(\rho) = -i[H, \rho],\ H=H^\dagger,\ H \in \Htr\} \cong \su(d).
\end{equation}
In particular, reversible dissipation is impossible: a generator with $\pm\Lcal \in \WGKLS$ has vanishing Kossakowski matrix.
\end{lemma}
\begin{theorem}[$\WGKLS$ is a Lie wedge \cite{dirr2009lie}]
\label{thm:gkls_lie_wedge}
The set $\WGKLS$ is a Lie wedge in the Lindblad Kossakowski Lie algebra $\g_{LK}$.
\end{theorem}
The dissipation functional and its consequences for commutators are collected in Appendix~\ref{app:wedge_construction}.

\subsection{Generator-Labelled Specifications and the DLWP}
\label{sec:dlwp_definition}
\begin{definition}[Dynamical System Specification]
A generator-labelled Markovian control specification is a finite family of independently actuable positive rays $\R_{\geq0}\Lcal_i$ with $\Lcal_i\in\WGKLS$. We write $G=\{\Lcal_1,\ldots,\Lcal_m\}$ for chosen representatives of those rays. Replacing a representative by a positive multiple does not change the specification. Bidirectional coherent control is stated by including both $\Lcal_H$ and $-\Lcal_H$ in $G$.
\end{definition}

\begin{definition}[Dynamical Lie Wedge Pair (DLWP)]
\label{def:DLWP}
Given a specification $G$, its \emph{Dynamical Lie Wedge Pair} is
\[
S(G)=(\g(G),W(G)).
\]
Here $\g(G)=\Lie(G)$ is the smallest real Lie subalgebra of $\g_{LK}$ containing $G$, and $W(G)$ is the smallest closed Lie wedge in $\g(G)$ containing $G$.
\end{definition}

\begin{remark}[Algebraic and reachability roles]
The algebra $\g(G)$ is the open-system counterpart of the dynamical Lie algebra of a control family \cite{jurdjevic1997geometric, d2021introduction, Schirmer2000Complete}. Proposition~\ref{prop:construction_W(G)_saturation} proves that $W(G)$ exists and gives a finite construction. The wedge is an inner approximation of the tangent wedge in Definition~\ref{def:W(G)_semigroup}. Equality holds exactly when $W(G)$ is global.
\end{remark}

\begin{remark}[What is being classified]
\label{rem:scope_control_systems}
The input is a generator-labelled specification of independently actuable rays, which determines the pair $(\g(G),W(G))$. If the only available dynamics is a fixed autonomous generator $\Lcal$, then $G=\{\Lcal\}$, the algebra $\g(G)=\R\Lcal$ is one-dimensional and abelian, the minimal local wedge is $\R_{\geq0}\Lcal$, and the global object is its semigroup. A Hamiltonian--jump decomposition may be used as $G$ only when the resulting rays are independently actuable, even if its components are physically distinguishable. The GKLS representation admits gauge transformations, including inhomogeneous shifts $L_k\mapsto L_k+c_k\mathbb I$ with a compensating change of $H$ and unitary recombinations of jump operators. Such changes do not create additional control rays.

The multi-generator case describes piecewise-controlled Markovian dynamics, in the same sense in which a dynamical Lie algebra describes a coherent control family.
\end{remark}

\paragraph{The Dynamical Lie Wedge $W(G)$.}
The wedge $W(G)$ is the minimal closed Lie-wedge local model determined by $G$, and by construction $W(G)\subseteq\g(G)$. Its inclusion in the GKLS cone follows from minimality. Indeed, $\WGKLS\cap\g(G)$ is a closed convex cone containing $G$, with edge $E(\WGKLS)\cap\g(G)$. For every $X$ in that edge, $e^{\ad X}$ preserves $\WGKLS$ by Theorem~\ref{thm:gkls_lie_wedge} and preserves $\g(G)$ because $\g(G)$ is a subalgebra. Hence $\WGKLS\cap\g(G)$ is a closed Lie wedge in $\g(G)$ containing $G$, and minimality gives
\begin{equation}
W(G) \subseteq \g(G) \cap \WGKLS.
\end{equation}
This conclusion is specific to the minimal wedge: an arbitrary Lie wedge containing $G$ need not lie in $\WGKLS$. For example, $\g(G)$ itself does so only when all its elements are Hamiltonian.
The geometry of $W(G)$ shows how positivity restricts the algebraic possibilities left open by $\g(G)$.

\paragraph{Local Wedge and Global Tangent Wedge.}
The global object associated with a switching family is its generated Lie semigroup and tangent wedge.
We use the invertible affine identity component $\mathcal A=\mathrm{GL}^{+}(\Htr)\ltimes\Htr$, where $\mathrm{GL}^{+}$ consists of matrices with positive determinant, with its finite-dimensional matrix topology. Closures below are taken in $\mathcal A$, not in the space containing singular infinite-time limits. An affine generator $(A,b)$ acts as $\dot x=Ax+b$ (thus $b=b_\Lcal/d$ in the representation above), with
\[
(A,b)\longleftrightarrow\begin{pmatrix}A&b\\0&0\end{pmatrix},
\qquad [(A,b),(A',b')]=([A,A'],Ab'-A'b).
\]
\begin{definition}[Lie Semigroup and Tangent Wedge]
\label{def:W(G)_semigroup}
The Lie semigroup $\Sigma(G)$ generated by $G$ is the closure in $\mathcal A$ of all finite products $e^{t_1 \Lcal_1}\cdots e^{t_k \Lcal_k}$ with $\Lcal_i\in G$ and $t_i\geq0$. Its tangent wedge is
\[
L(\Sigma(G))=\{X\in\g_{LK}:e^{tX}\in\Sigma(G)\text{ for every }t\geq0\}.
\]
\end{definition}
This is the exponential characterization of the tangent wedge \cite[Sec.~1.1, Eq.~(1.2)]{o2025lie}.

\begin{remark}[Switching interpretation]
The semigroup $\Sigma(G)$ is the closed reachable set for piecewise-constant switching among the declared rays. The tangent wedge $L(\Sigma(G))$ is a closed Lie wedge \cite{hilgert1989lie}.
\end{remark}

The tangent wedge contains $G$ but need not lie inside $\g(G)$, so the relevant comparison inside $\g(G)$ is $L(\Sigma(G)) \cap \g(G)$. It is a closed Lie wedge in $\g(G)$ containing $G$, by the argument already used for $\WGKLS \cap \g(G)$ in Section~\ref{sec:dlwp_definition}, with $L(\Sigma(G))$ in place of $\WGKLS$. Minimality of $W(G)$ therefore gives the inclusion
\begin{equation}\label{eq:saturate_vs_tangent}
W(G) \;\subseteq\; L(\Sigma(G)).
\end{equation}
The inclusion can be strict. If $G = \{\Lcal_H\}$ consists of a single Hamiltonian generator $\Lcal_H = -i[H,\cdot\,]$ with quasi-periodic flow, then $W(G) = \R_{\geq 0}\,\Lcal_H$ is already a closed Lie wedge with trivial edge, while the closure of $\{e^{t\Lcal_H} : t \geq 0\}$ is a compact group, so $L(\Sigma(G))$ contains $-\Lcal_H$ as well. Equality in \eqref{eq:saturate_vs_tangent} means that $W(G)$ is \emph{global} in the chosen ambient group $\mathcal A$ \cite{hilgert1989lie}. Theorem~\ref{thm:levi_type_main} in Appendix~\ref{app:levi_gauge_details} applies to arbitrary Lie wedges. Here we apply it to the minimal closed wedge used in the classification. Corollary~\ref{cor:qubit_globality_semialgebra} subsequently establishes globality for the classified family.
\subsection{Fibered Levi Data}
\label{sec:levi_type_theorem}
We consider a DLWP $(\g, W)$. The dimensions of the span and edge of the radical slice enter the classification tuple, while the quotient cone and fibers describe the wedge relative to a Levi splitting. Let $\g = \mf{s} \ltimes \mf{r}$ be the Levi decomposition of $\g$, where $\mf{r}$ is the solvable radical and $\mf{s}$ is the semisimple Levi factor. Let $\pi_{\s}: \g \to \mf{s} \cong \g/\mf{r}$ be the canonical projection homomorphism, with $\Ker(\pi_{\s}) = \mf{r}$. We identify $\g$ with the vector space $\mf{s} \oplus \mf{r}$ via the fixed Levi section. We adapt the toolkit of the mathematical theory of semigroups \cite{hilgert1989lie} to the physical constraints of GKLS generators.

\begin{definition}[Components of the decomposition]
\label{def:fiber_cone}
Let $(\g, W)$ be a DLWP as above.
\begin{enumerate}
    \item The Radical Wedge is the intersection of $W$ with the solvable radical,
    \begin{equation}
    W_{\rr} = W \cap \mf{r}.
    \end{equation}
    \item The Quotient Cone is the canonical projection of $W$ onto the Levi factor,
    \begin{equation}
    C_{\text{quot}} = \pi_{\s}(W).
    \end{equation}
    \item For $S \in C_{\text{quot}}$, the fiber $W(S)$ is the nonempty closed convex set of radical components that complete $S$ to an element of the wedge,
    \begin{equation}
    W(S) = \{R \in \mf{r} \mid S+R \in W\}.
    \end{equation}
    The set-valued assignment $S \mapsto W(S)$ is the Fiber Map.
\end{enumerate}
\end{definition}

Geometrically, $S + W(S)$ is the slice of $W$ over the base point $S$, the intersection of $W$ with the affine subspace $S+\mf{r}$. The wedge is recovered from its slices by the identity
\begin{equation}
\label{eq:fiber_union}
W = \bigcup_{S \in C_{\text{quot}}} (S + W(S)).
\end{equation}
Throughout, the fiber data are relative to a fixed Levi splitting, and the quotient cone need not be closed. Appendix~\ref{app:levi_gauge_details} collects the component properties, their compatibility condition, and the reconstruction theorem, including transport under a change of Levi section.

Write $\rad(\g)=\rr$ for the solvable radical and $Z(\g)=\{X\in\g:[X,Y]=0\text{ for every }Y\in\g\}$ for the center.
\begin{definition}[Four Levi types]
\label{def:four_levi_types}
A DLWP $(\g,W)$ with Levi decomposition $\g=\mf{s}\ltimes\mf{r}$ has \emph{solvable type} if $\mf{s}=\{0\}$, \emph{semisimple type} if $\mf{r}=\{0\}$, \emph{mixed reductive type} if both factors are nonzero and $\rad(\g)=Z(\g)$, and \emph{mixed non-reductive type} if both factors are nonzero and $\rad(\g)\neq Z(\g)$.
\end{definition}

\Needspace{8\baselineskip}
\section{Classification of Single-Channel Qubit Controls}
\label{sec:qubit_classification}

For the family $\mathcal{F}$, the coherent-control algebra and channel discriminants determine the Levi-data stratum. We first define these data and the equivalences used in the classification. Three exceptional strata arise on lower-dimensional channel loci.

\subsection{Normalization, the Data Tuple, and the Symmetry Group}

For $d = 2$ we work in the Bloch parameterization. Writing $\rho = \tfrac{1}{2}(\mathbb{I} + v \cdot \sigma)$ with $v \in \R^3$, every element of $\gLK(2) \cong \gl(3,\R) \levi \R^3$ (Theorem~\ref{thm:structure_glk}) acts as an affine vector field $\dot v = Qv + b$ with $Q \in \gl(3,\R)$ and $b \in \R^3$. A Hamiltonian generator with $H = \tfrac{1}{2}\, n \cdot \sigma$ acts as the rotation generator $J_n := (n \times \cdot\,)$, with no translation part. For axial control we now choose Bloch coordinates with the control axis along $\hat z$ and write $J_z := J_{\hat z}$. The principal frame of Section~\ref{sec:control_consequences} is instead tied to the channel.

Kossakowski matrices use the default Hilbert--Schmidt orthonormal basis $F_j=\sigma_j/\sqrt{2}$. Thus a jump $L=a\cdot\sigma$ has $K^{(F)}=2aa^\dagger$ and $\Tr K^{(F)}=2|a|^2$. In the raw Pauli basis the same generator has $K^{(\sigma)}=aa^\dagger$.

\begin{definition}[Single-channel qubit family]
\label{def:qubit_family}
Let $\mathcal{F}$ denote the set of specifications $G \subset \WGKLS(2)$ of the following form. The coherent part consists of control pairs $\pm \Lcal_{H_1}, \dots, \pm \Lcal_{H_k}$ ($k \geq 0$), and we write $\mf{e}_0 := \langle \Lcal_{H_1}, \dots, \Lcal_{H_k}\rangle_{\mathrm{Lie}}$ for the control algebra. The dissipative part is at most one jump dissipator $\Dcal_L(\rho) = L\rho L^\dagger - \tfrac{1}{2}\{L^\dagger L, \rho\}$ with $L = a \cdot \sigma$ traceless, $a \in \C^3 \setminus \{0\}$. We require at least one nonzero generator in $G$ and exclude drift terms.
\end{definition}

\begin{remark}[Admissible control algebras]
A two-dimensional span of rotation generators brackets onto a third independent axis. Hence the control algebra $\mf{e}_0$ in Definition~\ref{def:qubit_family} is one of $\{0\}$, $\so(2)_n$, or $\so(3)$.
\end{remark}

Within $\mathcal F$, we identify coherent lists that generate the same control algebra $\mf e_0$ and channels on the same positive dissipative ray. Lemma~\ref{lem:qubit_normalization}(iii) proves that these identifications preserve $(\g(G),W(G))$. Arbitrary independently specified generator lists retain the actuation distinctions described in Remark~\ref{rem:scope_control_systems}.

\begin{definition}[Levi-type data tuple]
For $G \in \mathcal{F}$ with pair $(\g, W) = (\g(G), W(G))$, Levi decomposition $\g = \s \levi \rr$, and edge $E = E(W)$, the \emph{Levi-type data tuple} is
\begin{equation}
T(G) \;=\; \bigl(\dim \g,\; \dim \rr,\; [\s],\; \mathrm{red},\; \mathrm{ab};\; \dim E,\; \dim (E \cap \rr),\; \dim \operatorname{span} \Wrad,\; \dim E(\Wrad),\; n_{\mathrm{tr}}\bigr),
\end{equation}
where $[\s]$ is the isomorphism class of the Levi factor, $\mathrm{red}$ and $\mathrm{ab}$ flag reductivity and commutativity of $\g$, $\Wrad = W \cap \rr$ is the radical wedge of Theorem~\ref{thm:levi_type_main}, and $n_{\mathrm{tr}} = \dim(\g \cap \R^3)$ counts the pure translations contained in the system algebra.
\end{definition}

Each entry of $T(G)$ is determined by the pair $(\g(G),W(G))$ inside the ambient pair $(\gLK,\WGKLS)$. This observation underlies the following symmetry and normalization lemma.

\begin{lemma}[Allowed equivalences preserve the Levi data]
\label{lem:qubit_normalization}
Let $G \in \mathcal{F}$.
\begin{enumerate}
    \item \textup{(Unitary covariance.)} For $u \in \mathrm{SU}(2)$ let $u \cdot G$ be the specification with Hamiltonians $uH_ju^\dagger$ and jump $uLu^\dagger$. Then $T(u \cdot G) = T(G)$. On Bloch data the action is conjugation by the rotation $R_u \in \mathrm{SO}(3)$.
    \item \textup{(Transpose extension.)} Conjugating every generator by the transpose map $\Theta(\rho) = \rho^{T}$ sends $\Dcal_L$ to $\Dcal_{\bar L}$ and $\Lcal_H$ to $\Lcal_{-\bar H}$, hence maps $\mathcal{F}$ to itself, preserves $\WGKLS$, and acts on Bloch data by the reflection $R_y = \mathrm{diag}(1,-1,1)$. The tuple is invariant, so the effective symmetry group extends from $\mathrm{SO}(3)$ to $\mathrm{O}(3)$.
    \item \textup{(Equivalent control lists and positive channel scaling.)} Replacing a bidirectional coherent generating list by another list with the same generated control algebra $\mf e_0$, or multiplying the single dissipative generator by a positive scalar, leaves $\g(G)$, $W(G)$, and $T(G)$ unchanged.
\end{enumerate}
\end{lemma}

The proof is given in Appendix~\ref{app:normalization_proof}.

Appendix~\ref{app:normalization_proof} also explains how inhomogeneous jump-gauge shifts can change a separated control specification and why the orthogonal group is the full group of linear Bloch symmetries of the GKLS wedge.

\subsection{The Channel Normal Form}
\label{sec:channel_normal_form}

The channel normal form separates the symmetric Bloch matrix from the translation. Resolving these data relative to the control axis produces the discriminants used in the classification.

\begin{lemma}[Bloch data of a traceless jump]
\label{lem:single_jump_normal_form}
Write $L = a \cdot \sigma$ with $a = \alpha + i\beta$, $\alpha, \beta \in \R^3$. Then $\Dcal_L$ acts as $\dot v = Qv + b$ with
\begin{equation}
\label{eq:jump_normal_form}
Q \;=\; 2\bigl(\alpha\alpha^{T} + \beta\beta^{T} - |a|^2\,\mathbb{I}\bigr),
\qquad
b \;=\; 4\,(\alpha \times \beta).
\end{equation}
In particular $Q$ is symmetric and negative semidefinite, and a single traceless channel carries no intrinsic rotation. Every Hamiltonian direction in $\g(G)$ beyond the controls is generated algebraically by brackets.
\end{lemma}

\begin{proof}
Expand $\Dcal_L$ in the Pauli basis with Kossakowski matrix $K = 2a\bar a^{T}$. The action on $\sigma_l$ contributes no identity component, since the two $\epsilon$-terms produced by $\sigma_j \sigma_l \sigma_k$ and $\{\sigma_k\sigma_j, \sigma_l\}$ cancel by cyclicity. The surviving matrix part is built from $\operatorname{Re} K$ alone. The imaginary part of $K$ drops out of $Q$ because the matrix part symmetrizes $K$ over $j \leftrightarrow k$. It reappears as the translation $b$ through $\Dcal_L(\mathbb{I}) = \tfrac{1}{2} K_{jk}[\sigma_j, \sigma_k]$, with $\epsilon_{jkm} \operatorname{Im} K_{jk} \propto (\alpha \times \beta)_m$. Collecting constants in the normalization $\rho = \tfrac{1}{2}(\mathbb{I} + v\cdot\sigma)$ yields \eqref{eq:jump_normal_form}.
\end{proof}

\begin{corollary}[Principal form and stabilizer]
\label{cor:single_jump_principal_form}
After a phase rotation of $a$, one may take $\alpha\perp\beta$ with $s:=|\alpha|\geq t:=|\beta|$ and $\hat\alpha=\alpha/s$. For $t>0$, set $\hat\beta=\beta/t$ and $\hat\nu=\hat\alpha\times\hat\beta$. For $t=0$, choose any right-handed orthonormal completion of $\hat\alpha$. In this principal frame, the eigenvalues of $Q$ are $-2t^2$, $-2s^2$, and $-2(s^2+t^2)$, and $b=4st\,\hat\nu$. The channel is unital exactly when $t=0$. Its Bloch data have a continuous rotation axis exactly when $t=0$ or $t=s$. For $0<t<s$, the only nontrivial proper rotation fixing them is the half-turn about $\hat\nu$.
\end{corollary}

\begin{proof}
The phase rotation $a \mapsto e^{i\varphi} a$ leaves $\Dcal_L$ fixed and rotates the pair $(\alpha, \beta)$ inside its plane, so principal axes exist. Substitution in \eqref{eq:jump_normal_form} gives the eigenvalues and translation. The stabilizer under proper Bloch rotations in $\mathrm{SO}(3)$ follows by inspecting the eigenframe of $Q$ together with $b$: nondegeneracy for $0<t<s$ leaves only the half-turn about $\hat\nu$, while a full rotation axis occurs in the two stated cases.
\end{proof}

The normalized channel parameter in Section~\ref{sec:control_consequences} is $q=t/s$. Positive channel scaling sets $s=1$ there. We retain $s$ and $t$ below to express Hamiltonian strengths and convergence rates.

To identify the axial branches, decompose the matrix $M := \alpha\alpha^{T} + \beta\beta^{T}$, equivalently $Q$, relative to the control axis. Under the rotations $e^{tJ_z}$ the symmetric matrices split into weight sectors, and we write $m_1(M) := (M_{xz}, M_{yz})$ for the weight-one component, $m_2(M)$ for the traceless block of the restriction to the $xy$-plane (weight two), and
\begin{equation}
\label{eq:dev_def}
\operatorname{dev}(M) := M_{zz} - \tfrac{1}{3}\operatorname{tr} M
\end{equation}
for the deviatoric part of the weight-zero component. The vector $b$ splits into $b_z$ and the planar part $b_\perp$. The four channel data $(m_1,m_2,\operatorname{dev},b)$ provide the classification discriminants. The condition $b=0$ is the unitality test.
The component values are distinct from their matrix sectors. With $E_{ij}$ the matrix units, set
\[
\mathcal S_1=\Span\{E_{13}+E_{31},E_{23}+E_{32}\},\quad
\mathcal S_2=\Span\{E_{11}-E_{22},E_{12}+E_{21}\},\quad
\mathcal S_0=\R\operatorname{diag}(1,1,-2).
\]
The pair $m_1(M)$ embeds as $M_{xz}(E_{13}+E_{31})+M_{yz}(E_{23}+E_{32})$ in $\mathcal S_1$. The block $m_2(M)$ embeds in the upper-left $2\times2$ block of a $3\times3$ matrix, with all other entries zero. Writing $d_0=\operatorname{dev}(M)$, the weight-zero traceless projection of $M$ is $\operatorname{diag}(-d_0/2,-d_0/2,d_0)\in\mathcal S_0$. These embeddings give
\[
\operatorname{Sym}(3,\R)=\R I\oplus\mathcal S_0\oplus\mathcal S_1\oplus\mathcal S_2.
\]
The weight labels describe rotation frequency about $\hat z$: the $\mathcal S_1$ and $\mathcal S_2$ components rotate through angles $\theta$ and $2\theta$, respectively, while $\mathcal S_0$ is fixed.

\Needspace{8\baselineskip}
\subsection{Edge Rigidity}

For any finite gauge-normalized single-qubit specification, the dissipative generators leave the reversible part of the wedge unchanged. The following lemma derives this constraint from the facial geometry of the GKLS cone, including specifications with several channels or one-sided drift terms.

Here gauge normalization refers to the traceless-jump GKLS representation. Controls are the declared bidirectional Hamiltonian directions, while Hamiltonian drift directions are available only with one sign.

\begin{lemma}[Edge rigidity]
\label{lem:qubit_edge_rigidity}
Let $G$ be any finite gauge-normalized single-qubit specification, possibly with several channels and one-sided drift terms.
\begin{enumerate}
    \item Every element of $E(W(G))$ is a Hamiltonian generator, and $\dim E(W(G)) \in \{0, 1, 3\}$.
    \item Let $G_H \subseteq G$ be the subset of purely Hamiltonian generators. Then $E(W(G)) = E(W(G_H))$. Generators with a nonvanishing dissipative part are edge-inert.
    \item Let $C$ be the closed conic hull of the Hamiltonian drift directions, and let $\mf{e}$ be the Lie algebra generated by the controls together with the lineality space $C\cap(-C)$ of $C$. If $\mf{e} = \R J_n$ for some axis $n$ and some Hamiltonian drift axis is not parallel to $n$, then $E(W(G)) = \so(3)$. In every other case $E(W(G)) = \mf{e}$. For $G \in \mathcal{F}$ this reduces to $E(W(G)) = \mf{e}_0$.
\end{enumerate}
\end{lemma}

The proof is given in Appendix~\ref{app:edge_proof}.

The proof in Appendix~\ref{app:edge_proof} also identifies the half-turn mechanism that completes a misaligned Hamiltonian drift into the reversible edge. For the drift-free family considered here, the edge remains the declared coherent-control algebra.

\Needspace{8\baselineskip}
\subsection{Possible Levi Factors}

\begin{proposition}[Semisimple-type menu]
For every single-qubit specification $G$, of any size and with or without drift, the Levi factor of $\g(G)$ is isomorphic to one of
\begin{equation}
\{0\}, \qquad \sll(2,\R), \qquad \so(3), \qquad \sll(3,\R).
\end{equation}
\end{proposition}

The proof is given in Appendix~\ref{app:levi_menu_proof}.

All four types are realized already inside $\mathcal{F}$, and $\sll(2,\R)$ is realized in two ways inequivalent at the level of the pair $(\g, \s)$, reductively in stratum $\mathsf{Q4}$ and non-reductively in $\mathsf{Q5}$ and $\mathsf{Q6}$.

\subsection{The Classification Theorem}
\label{sec:qubit_strata}

\begin{theorem}[Exhaustive thirteen-value Levi-data stratification for single-channel qubits]
\label{thm:qubit_classification}
Let $G$ range over the family $\mathcal{F}$ of Definition~\ref{def:qubit_family}, taken up to the equivalence of Lemma~\ref{lem:qubit_normalization}.
\begin{enumerate}
    \item \textup{(Edge.)} The edge satisfies $E(W(G)) = \mf{e}_0$. It has dimension $0$, $1$, or $3$ and consists of Hamiltonian generators.
    \item \textup{(Tuple-value and discriminant completeness.)} The tuple $T(G)$ takes exactly the thirteen values listed in Table~\ref{tab:qubit_strata}. The control algebra together with the channel discriminants $\bigl(m_1,m_2,\operatorname{dev},b\bigr)$ of Section~\ref{sec:channel_normal_form} assigns every specification to exactly one stratum.
    \item \textup{(Realization.)} Every row is realized.
\end{enumerate}
\end{theorem}

\begin{remark}[Completeness scope]
\label{rem:qubit_completeness_scope}
The tuple is complete for stratum assignment. Finer equivalence questions also depend on $C_{\mathrm{quot}}$, the full fiber map, and the in-stratum channel moduli, so one tuple may represent several control specifications or DLWP isomorphism classes. In the moduli display, $\mathsf{Q3}$ has two isolated channel orbits. The strata $\mathsf{Q6}$, $\mathsf{Q8}$, and $\mathsf{Q12}$ each have one. The strata $\mathsf{Q2}$ and $\mathsf{Q11}$ are the single jump-free coherent-data orbits.
\end{remark}

\paragraph{Structure of the classification proof.}
The proof divides according to the coherent-control algebra. Without controls, a single channel generates a ray. Jump-free specifications give the coherent algebra itself. With one control axis, the weight decomposition separates the axisymmetric and weight-two branches from those with a nonzero weight-one component. For the branches with $m_1\neq0$, cross-weight brackets regenerate the full rotation algebra except on the $\mathsf{Q6}$ degeneration. Full coherent control generates the rotation orbit directly.

The affine lift is resolved by its pure-translation kernel $T=\{b:(0,b)\in\g(G)\}$. If all fields share a fixed point $v_\ast$, their translations satisfy $b=-Av_\ast$, so the algebra is a graph over its matrix projection. In the irreducible non-unital branches whose matrix projection contains the identity, absence of such a point gives all three translations by the affine-kernel dichotomy. The matrix projection and translation kernel determine the radical, while orbit averages and the magic condition decide which dissipative directions lie in it. Appendix~\ref{app:qubit_classification_proofs} gives the full parameter-domain calculations, together with the realization and residual-moduli arguments.

Table~\ref{tab:qubit_strata} gives the thirteen tuple values and the tests that select them. In its witnesses, $\sigma_-=(\sigma_x-i\sigma_y)/2=|1\rangle\langle0|$ in the computational basis.

\begin{table}[H]
\centering
\footnotesize
\setlength{\tabcolsep}{3pt}
\begin{tabular}{c c l c c c c}
\hline
stratum & $\mf{e}_0$ & channel class (discriminants) & $(\dim\g, \dim\rr, [\s])$ & type & $(E, E\cap\rr)$ & $(\mathrm{sp}\,\Wrad, E(\Wrad), n_{\mathrm{tr}})$ \\
\hline
$\mathsf{Q1}$ & $0$ & any single channel & $(1,1,0)$ & sv, ab & $(0,0)$ & $(1,0,0)$ \\
$\mathsf{Q2}$ & $\so(2)$ & none & $(1,1,0)$ & sv, ab & $(1,1)$ & $(1,1,0)$ \\
$\mathsf{Q3}$ & $\so(2)$ & axisymmetric, $m_1 = m_2 = 0$ & $(2,2,0)$ & sv, ab & $(1,1)$ & $(2,1,0)$ \\
$\mathsf{Q4}$ & $\so(2)$ & $m_1 = 0$, $m_2 \neq 0$, $b \in \R\hat z$ & $(4,1,\sll(2,\R))$ & red & $(1,0)$ & $(1,0,0)$ \\
$\mathsf{Q5}$ & $\so(2)$ & $m_1 = 0$, $m_2 \neq 0$, $b_\perp \neq 0$ & $(6,3,\sll(2,\R))$ & n-red & $(1,0)$ & $(3,0,2)$ \\
$\mathsf{Q6}$ & $\so(2)$ & $m_1 \neq 0$, $m_2 = 0$, $\operatorname{dev} = 0$ & $(7,4,\sll(2,\R))$ & n-red & $(1,0)$ & $(1,0,3)$ \\
$\mathsf{Q7}$ & $\so(2)$ & unital, $m_1 \neq 0$, $\operatorname{dev} \neq 0$ & $(9,1,\sll(3,\R))$ & red & $(1,0)$ & $(0,0,0)$ \\
$\mathsf{Q8}$ & $\so(2)$ & unital, $m_1 \neq 0$, $\operatorname{dev} = 0$ & $(9,1,\sll(3,\R))$ & red & $(1,0)$ & $(1,0,0)$ \\
$\mathsf{Q9}$ & $\so(2)$ & non-unital, $m_1 \neq 0$, $\operatorname{dev} \neq 0$ & $(12,4,\sll(3,\R))$ & n-red & $(1,0)$ & $(0,0,3)$ \\
$\mathsf{Q10}$ & $\so(2)$ & non-unital, $m_1 \neq 0$, $m_2 \neq 0$, $\operatorname{dev} = 0$ & $(12,4,\sll(3,\R))$ & n-red & $(1,0)$ & $(1,0,3)$ \\
$\mathsf{Q11}$ & $\so(3)$ & none & $(3,0,\so(3))$ & ss & $(3,0)$ & $(0,0,0)$ \\
$\mathsf{Q12}$ & $\so(3)$ & unital & $(9,1,\sll(3,\R))$ & red & $(3,0)$ & $(1,0,0)$ \\
$\mathsf{Q13}$ & $\so(3)$ & non-unital & $(12,4,\sll(3,\R))$ & n-red & $(3,0)$ & $(4,0,3)$ \\
\hline
\end{tabular}
\caption{The thirteen single-channel qubit strata. Abbreviations: sv solvable, ab abelian, red reductive, n-red non-reductive, ss semisimple. The notation $\mathrm{sp}$ denotes the linear span, and the two rightmost column groups list dimensions. Witnesses are dephasing $L = \sigma_m/\sqrt{2}$ along the axis the discriminants dictate, extremal damping $L = \sigma_-$ conjugated to the dictated axis, and for $\mathsf{Q6}$ the channel $a = (1, i, 1)$. The magic locus $\operatorname{dev}(M) = 0$ is the condition $|a_z|^2 = |a|^2/3$ of Lemma~\ref{lem:magic_locus}. Lemma~\ref{lem:qubit_parametric_affine_closure} and Corollaries~\ref{cor:qubit_parametric_radicals}--\ref{cor:qubit_parametric_radical_wedges} prove the affine tuple on every branch.}
\label{tab:qubit_strata}
\end{table}

The residual data below describe the channel orbits that remain within each stratum after the stated equivalences. Write $q=t/s\in[0,1]$ in the principal form of Corollary~\ref{cor:single_jump_principal_form}, after positive channel scaling. For $\so(2)_{\hat z}$ control, let $[\hat z]_{\mathrm{pr}}$ denote the orbit of the unoriented control-axis line in the principal frame under the orthogonal channel stabilizer $\{R\in\mathrm O(3):RQR^T=Q,\ Rb=b\}$. When useful, set
\[
x=|\hat z\cdot\hat\alpha|^2,\qquad
y=|\hat z\cdot\hat\beta|^2,\qquad
\zeta=|\hat z\cdot\hat\nu|^2,\qquad x+y+\zeta=1.
\]
The magic equation is
\[
x+q^2y=\frac{1+q^2}{3}.
\]

\begin{center}
\footnotesize
\begin{tabular}{c p{0.79\textwidth}}
\hline
stratum & residual data after the stated equivalences \\
\hline
$\mathsf{Q1}$ & $q\in[0,1]$. \\
$\mathsf{Q2}$ & one jump-free $\so(2)$ coherent-data orbit. \\
$\mathsf{Q3}$ & two isolated orbits: axial dephasing ($q=0$) and axial extremal damping ($q=1$). \\
$\mathsf{Q4}$ & $q\in[0,1)$. At $q=0$ the dephasing axis is transverse to $\hat z$, while for $0<q<1$ one has $\hat\nu\parallel\hat z$. \\
$\mathsf{Q5}$ & for $0<q<1$, two branches $\hat z=\hat\alpha$ and $\hat z=\hat\beta$. They merge at $q=1$. \\
$\mathsf{Q6}$ & one orbit, with $q=1/\sqrt2$ and $(x,y,\zeta)=(1/2,0,1/2)$ up to the stated equivalences. \\
$\mathsf{Q7}$ & the unital tilt $x\in(0,1)\setminus\{1/3\}$. \\
$\mathsf{Q8}$ & one unital magic-tilt orbit, $x=1/3$. \\
$\mathsf{Q9}$ & the continuous set of $(q,[\hat z]_{\mathrm{pr}})$ with $q\in(0,1]$, $m_1\neq0$, and $\operatorname{dev}\neq0$. \\
$\mathsf{Q10}$ & the continuous locus of $(q,[\hat z]_{\mathrm{pr}})$ with $q\in(0,1]$ satisfying the magic equation above, $m_1\neq0$, and $m_2\neq0$. The $\mathsf{Q6}$ degeneration is excluded. \\
$\mathsf{Q11}$ & one jump-free $\so(3)$ coherent-data orbit. \\
$\mathsf{Q12}$ & one full-control dephasing orbit, $q=0$. \\
$\mathsf{Q13}$ & $q\in(0,1]$. \\
\hline
\end{tabular}
\end{center}

The possible dimensions of $\g(G)$ are $\{1,2,3,4,6,7,9,12\}$. The exceptional strata $\mathsf{Q6}$, $\mathsf{Q8}$, and $\mathsf{Q10}$ lie on measure-zero loci selected by vanishing conditions on the isotypic discriminants.

\subsection{Orbit Geometry and Radical Wedges}
\label{sec:orbit_geometry}
The algebraic classification determines the radical, but computing its intersection with the wedge requires the nonnegative weights of dissipative directions. Compact-control orbit averages retain those weights and reveal a cone discriminant that Lie closure alone does not detect.

Let $G \in \mathcal{F}$ have control algebra $\mf{e}_0 \neq 0$ and a channel $\Dcal_L$ with Bloch data \eqref{eq:jump_normal_form}. The controls generate a compact connected group $\mathcal{K} \subseteq \mathrm{SO}(3)$ of Bloch rotations, either a circle group or all of $\mathrm{SO}(3)$. For $R \in \mathcal{K}$, the rotated channel $\operatorname{Ad}_R \Dcal_L$ has Bloch data $(RQR^{T}, Rb)$. We describe nonnegative combinations of these orientations by a finite nonnegative Borel measure $\mu$ on $\mathcal{K}$ and set
\begin{equation*}
A_\mu \;:=\; \int_{\mathcal{K}} \operatorname{Ad}_R(\Dcal_L)\; d\mu(R),
\qquad
\lambda(\mu) \;:=\; \mu(\mathcal{K}).
\end{equation*}

\begin{lemma}[Orbit-average representation]
\label{lem:qubit_orbit_average}
The wedge $W(G)$ is the closed conic hull of
\[
\mf{e}_0 \cup \{\operatorname{Ad}_R \Dcal_L : R \in \mathcal{K}\}.
\]
If no element of $\rr$ has a nonzero rotation component, as in every nonabelian stratum of Table~\ref{tab:qubit_strata}, then
\[
\Wrad=\{A_\mu:\mu\geq0,\ A_\mu\in\rr\},
\]
and this set is closed. The matrix part of $A_\mu$ has trace $-4|a|^2\lambda(\mu)$. Hence $\Wrad$ is pointed and $E(\Wrad)=0$.
\end{lemma}

\begin{proof}
Write $C$ for the stated conic hull and $\varphi$ for the Kossakowski-trace functional from the proof of Lemma~\ref{lem:qubit_edge_rigidity}. This functional is linear, nonnegative on $\WGKLS$, invariant under rotations, zero on $\mf{e}_0$, and equal to $2|a|^2\lambda(\mu)$ on $A_\mu$. Elements of $C$ are limits of sums $Y_n + A_{\mu_n}$ with $Y_n \in \mf{e}_0$ and $\mu_n$ finitely supported. If such a limit $X$ has $\varphi(X) = 0$, then $\lambda(\mu_n) \to 0$, the boundedness of the orbit $\{\operatorname{Ad}_R \Dcal_L\}$ gives $A_{\mu_n} \to 0$, and $X \in \mf{e}_0$. Edge elements come with both signs, so $\varphi$ vanishes on $E(C)$ and $E(C) = \mf{e}_0$. The cone $C$ is invariant under $e^{\operatorname{ad} \mf{e}_0}$ because its generating set is. So $C$ is a closed Lie wedge containing $\operatorname{cone}(G)$. Conversely, any closed Lie wedge containing $\operatorname{cone}(G)$ contains the controls with both signs, hence contains $\mf{e}_0$ in its edge, edges of Lie wedges being subalgebras. Invariance under $e^{\operatorname{ad} \mf{e}_0}$ follows from the wedge property, passes to the generated group $\mathcal{K}$ by closedness, and delivers the orbit and all of $C$. Minimality of $W(G)$ gives $W(G) = C$.

Now let $X \in W(G) \cap \rr$ and write $X = \lim\,(Y_n + A_{\mu_n})$. Rotated channels carry symmetric matrix parts and translations and no rotation component (Lemma~\ref{lem:single_jump_normal_form}), so the rotation component of the approximants is $Y_n$, and it converges to the rotation component of $X$, which vanishes by hypothesis. Hence $Y_n \to 0$ and $A_{\mu_n} \to X$. The masses converge because $\varphi(A_{\mu_n}) \to \varphi(X)$, so along a subsequence the $\mu_n$ converge weakly-$\ast$ to a finite $\mu \geq 0$, by compactness of bounded measure sets on the compact group $\mathcal{K}$, and $X = A_\mu$ by continuity of $\mu \mapsto A_\mu$. The same compactness shows the set of admissible averages is closed. Conversely, every $A_\mu$ lies in $W(G)$: finitely supported measures give conic combinations of orbit elements, and a general $\mu$ is a weak-$\ast$ limit of finitely supported measures of the same mass, so closedness of $W(G)$ finishes.

The trace of the matrix part is $\int \operatorname{tr}(RQR^{T})\, d\mu = \lambda(\mu) \operatorname{tr} Q = -4|a|^2 \lambda(\mu)$ by \eqref{eq:jump_normal_form}. If $A_\mu$ and $-A_\mu = A_{\mu'}$ both lie in $\Wrad$, then $A_{\mu+\mu'} = 0$ forces $\lambda(\mu) + \lambda(\mu') = 0$, so both measures are zero.
\end{proof}

For full coherent control, radical membership requires the averaged matrix part to be scalar. The next corollary expresses this as the vanishing of its traceless second moment, or quadrupole moment. The averaged translation is the corresponding dipole moment.

\begin{corollary}[Full-control radical wedge]
\label{cor:qubit_full_control_radical_wedge}
Suppose $\mf{e}_0=\so(3)$ and $\rr=\R\mathbb I\oplus T$, where $T=0$ in the unital case and $T=\R^3$ in the non-unital case. Then $A_\mu\in\rr$ exactly when
\[
\int RMR^{T}\,d\mu-\tfrac13\lambda(\mu)(\operatorname{tr}M)\mathbb I=0.
\]
For these quadrupole-free measures, the matrix part of $A_\mu$ equals $-\tfrac43|a|^2\lambda(\mu)\mathbb I$. Their dipole moments $\int Rb\,d\mu$ realize every vector in $\R^3$ when $b\neq0$. Hence $\operatorname{span}\Wrad=\R\mathbb I\oplus\R^3$ for a non-unital channel and $\operatorname{span}\Wrad=\R\mathbb I$ for a unital one.
\end{corollary}

\begin{proof}
The matrix part of $A_\mu$ is $2\int RMR^{T} d\mu - 2|a|^2 \lambda(\mu)\,\mathbb{I}$, and in both full-control strata every translation of $\g$ lies in $\rr$, so membership in $\rr$ amounts to the matrix part lying in $\R\mathbb{I}$, the stated quadrupole condition. The coefficient follows from the trace formula in Lemma~\ref{lem:qubit_orbit_average}. With $dR$ the Haar probability measure, the uniform average is quadrupole-free with zero dipole moment, since $\int RMR^{T} dR = \tfrac{1}{3}(\operatorname{tr} M)\,\mathbb{I}$ by Schur orthogonality. When $b\neq0$, set $\hat b=b/|b|$. To prescribe a dipole along a unit vector $u$, take $d\mu_\delta = (1 + \delta\, u \cdot R\hat b)\, dR$, which is nonnegative for $|\delta| \leq 1$ and has dipole moment $\delta \int (u \cdot R\hat b)\, Rb\, dR = \tfrac{1}{3}\delta |b|\, u$. Its quadrupole correction is a cubic moment of the rotation entries, and the only isotropic rank-three tensor is the Levi-Civita symbol, which contracts to zero against the symmetric matrix $M$. Positive scaling of these measures lets the dipoles sweep $\R^3$ when $b \neq 0$, while $b = 0$ leaves the ray of the Haar average.
\end{proof}

For axial control, define the unit-mass channel average
\begin{equation}
\label{eq:axial_unit_mass_average}
w_0=\frac{1}{2\pi}\int_0^{2\pi}\Ad_{e^{\theta J_z}}\Dcal_L\,d\theta.
\end{equation}
It is the weight-zero affine component, with strictly negative matrix trace. A measure of mass $\lambda$ contributes exactly $\lambda w_0$ in this sector.

\begin{corollary}[Axial circular radical wedge]
\label{cor:qubit_axial_radical_wedge}
Suppose $\mf{e}_0=\so(2)_{\hat z}$, $m_1=0$, $m_2\neq0$, and $b_\perp\neq0$, and assume $\rr=\R w_0\oplus\hat z^\perp$. Identify $\hat z^\perp$ with $\C$ so that $e^{\theta J_z}$ acts by multiplication by $e^{i\theta}$, and set $m=\int e^{i\theta}\,d\mu(\theta)$. Then $A_\mu\in\rr$ exactly when $\int e^{2i\theta}\,d\mu=0$, and
\begin{equation}
\label{eq:q5_cone}
\Wrad=\{\,\lambda w_0+m b_\perp:\lambda\geq0,\ m\in\C,\ |m|\leq\lambda/\sqrt2\,\}.
\end{equation}
The constant $1/\sqrt2$ is sharp, and $\operatorname{span}\Wrad=\rr$ is three-dimensional.
\end{corollary}

\begin{proof}
Decompose the rotated channel against the weight sectors of Section~\ref{sec:channel_normal_form}. The weight-zero component of $A_\mu$ is $\lambda(\mu)\, w_0$. The weight-one matrix sector is absent because $m_1 = 0$. The planar translations rotate with $e^{i\theta}$ and contribute $m\, b_\perp$, inside $\rr$. The only sector transverse to $\rr$ is the weight-two block, which averages to the second Fourier moment acting on $m_2 \neq 0$. Membership is therefore the vanishing of that moment, and the average takes the displayed form. For the moment range, the constraint $\int \cos 2\theta\, d\mu = 0$ gives $\int \cos^2\theta\, d\mu = \lambda/2$, and the Cauchy--Schwarz inequality bounds $\int \cos\theta\, d\mu \leq \sqrt{\lambda}\,\sqrt{\lambda/2} = \lambda/\sqrt{2}$. Rotating $\theta$ moves the bound to every direction of $m$, so $|m| \leq \lambda/\sqrt{2}$. The two-point measure $\tfrac{\lambda}{2}\bigl(\delta_{\theta_0 - \pi/4} + \delta_{\theta_0 + \pi/4}\bigr)$ has vanishing second moment and attains $m = (\lambda/\sqrt{2})\, e^{i\theta_0}$, and convex combinations fill the disk. The span statement follows since $w_0 \neq 0$, its trace part being strictly negative, and $b_\perp \neq 0$.
\end{proof}

In the nine- and twelve-dimensional strata, the weight-zero sector can itself obstruct membership in the radical. The following criterion also applies to the exceptional seven-dimensional stratum. The next lemma identifies the channel condition that removes this obstruction and separates the corresponding rows of Table~\ref{tab:qubit_strata}.

\begin{lemma}[Magic locus]
\label{lem:magic_locus}
Let $G \in \mathcal{F}$ have control algebra $\so(2)_{\hat z}$ and a channel with $m_1(M) \neq 0$. Assume that the radical is $\rr = \R\mathbb{I} \oplus (\g(G) \cap \R^3)$. Recall from \eqref{eq:dev_def} that $\operatorname{dev}(M) = M_{zz} - \tfrac{1}{3}\operatorname{tr}M$ is the axial deviator. Then
\begin{equation}
\label{eq:magic_locus}
W(G) \cap \rr \neq \{0\}
\quad\Longleftrightarrow\quad
\operatorname{dev}(M) = 0
\quad\Longleftrightarrow\quad
|a_z|^2 = \tfrac{1}{3}|a|^2 ,
\end{equation}
and on the locus the intersection is the single ray generated by the orbit average of $\Dcal_L$. The ray is pointed. Its generator has strictly negative trace part, so $E(W \cap \rr) = 0$.
\end{lemma}

\begin{proof}
By Lemma~\ref{lem:qubit_orbit_average}, $W(G) \cap \rr$ consists exactly of the channel averages $A_\mu$ that lie in $\rr$. Both $\operatorname{tr} Q$ and $Q_{zz}$ are constant along the orbit. An average of mass $\lambda > 0$ therefore has matrix part with trace $-4|a|^2\lambda$ and $zz$-entry $2(|a_z|^2 - |a|^2)\lambda$, and it lies in $\rr$ only if the traceless symmetric part vanishes, which forces the ratio condition $Q_{zz} = \tfrac{1}{3}\operatorname{tr}Q$, equivalent to $\operatorname{dev}(M) = 0$ and to $|a_z|^2 = |a|^2/3$. For sufficiency, average the orbit against the uniform weight. Every nonzero weight component integrates to zero, and what remains is $\lambda(c\,\mathbb{I} + b_z \hat z)$ with $c = -\tfrac{4}{3}|a|^2 < 0$, an element of $\rr$. Any admissible weight must annihilate every Fourier moment belonging to a nonzero orbit sector: the first moment in the $\mathsf{Q6}$ degeneration, and the first and second moments in the other magic branches. The resulting element is therefore fixed by the mass $\lambda$ alone, so the intersection is one ray. Pointedness follows from the trace argument in Lemma~\ref{lem:qubit_orbit_average}.
\end{proof}

We call $\{\operatorname{dev}(M) = 0\}$ the \emph{magic locus}. The tilt angle it selects for an axially symmetric channel satisfies $\sin^2\theta^\ast = 2/3$, the magic angle of magnetic resonance.

\subsection{Globality}

We now connect the minimal wedge to the closed reachable semigroup in the ambient group $\mathcal A$. The coherent-control groups in the classified family are the trivial group, $\mathrm{SO}(2)$, and $\mathrm{SO}(3)$, all of which are closed. The proof uses this closedness to identify the tangent wedge and then determines when the stronger Lie-semialgebra property holds.

\begin{corollary}[Globality and the semialgebra boundary]
\label{cor:qubit_globality_semialgebra}
For every $G\in\mathcal F$, the classified wedge is global:
\[
W(G)=L(\Sigma(G)).
\]
It is a Lie semialgebra exactly in strata $\mathsf{Q1}$, $\mathsf{Q2}$,
$\mathsf{Q3}$, and $\mathsf{Q11}$.
\end{corollary}

\begin{proof}
The jump-free strata are immediate: their bidirectional coherent rays generate
the closed group $K$, so $\Sigma(G)=K$ and
$L(\Sigma(G))=\mf e_0=W(G)$. Suppose now that the jump dissipator $D$ is
present. Let $\Sigma_{\mathrm{sep}}=\Sigma(G)$ be the closed semigroup generated
by the independently actuated rays, and let $\Sigma_{\mathrm{dc}}$ be the closed
reachable semigroup for the drift-plus-control model used by O'Meara. The
bidirectional coherent rays give $K\subseteq\Sigma_{\mathrm{sep}}$. Every
admissible constant-control arc also belongs to $\Sigma_{\mathrm{sep}}$,
because, for $H$ in the actuated coherent span,
\[
e^{t(D+H)}=\lim_{m\to\infty}
\left(e^{tD/m}e^{tH/m}\right)^m.
\]
Thus $\Sigma_{\mathrm{dc}}\subseteq\Sigma_{\mathrm{sep}}$. Conversely,
$e^{tD}$ is obtained by setting the coherent control to zero, while
the unbounded bidirectional controls give the strong-pulse limit
\[
e^{tH}=\lim_{M\to\infty}e^{(t/M)(D+MH)}
\]
for every directly actuated coherent direction $H$. Products of these
one-parameter subgroups and closure generate $K$, proving
$\Sigma_{\mathrm{sep}}\subseteq\Sigma_{\mathrm{dc}}$ and hence equality. The
control groups in the classified family are closed.

To apply the closed-control theorem in the chosen ambient group, we also check that system-relative and ambient closures agree. The connected system groups with Lie algebra $\g(G)$ are closed in $\mathcal A$ as well. To see this, use the projections and kernels in Lemma~\ref{lem:qubit_parametric_affine_closure}. In $\mathsf{Q1}$ and $\mathsf{Q3}$ the determinant fixes the dissipative parameter, since $\operatorname{tr}Q<0$, and the remaining rotation factor is compact. In $\mathsf{Q4}$ and $\mathsf{Q5}$ the linear group is the embedded $\mathrm{SL}(2,\R)$ times $\exp(\R C)$, with $C$ the commuting weight-zero matrix from that lemma. Its nonzero trace again fixes the real parameter. The common-fixed-point graph in $\mathsf{Q4}$ and the planar translation extension in $\mathsf{Q5}$ preserve closedness. In $\mathsf{Q6}$ the irreducible semisimple image is conjugate to $\mathrm{SO}_0(2,1)$, the identity component of the indefinite orthogonal group, together with scalar dilation and all translations. The other channel rows have linear group $\mathrm{GL}^{+}(3,\R)$ with zero or full translations. The jump-free groups are compact. Each factor is closed, and a convergent invertible sequence has a convergent dilation parameter determined by its determinant. Thus ambient closure agrees here with closure relative to the system group.

O'Meara's closed-control theorem
\cite[Thm.~13]{o2025lie} gives the orbit-cone formula of
Lemma~\ref{lem:qubit_orbit_average} and proves globality. The semialgebra
criterion \cite[Thm.~15]{o2025lie} is commutation of the control algebra with
the dissipative generator. This condition is vacuous in $\mathsf{Q1}$, holds
for the jump-free rows $\mathsf{Q2}$ and $\mathsf{Q11}$, and holds for the two
axisymmetric channels in $\mathsf{Q3}$. In rows $\mathsf{Q4}$--$\mathsf{Q10}$,
at least one nonzero weight component changes under the axial control. Under
full $\mathrm{SO}(3)$ invariance one would need $Q$ to be scalar and $b=0$,
equivalently $M$ to be scalar. Since a single-jump matrix $M$ has rank at most
two, the only scalar possibility is the zero channel, which excludes
$\mathsf{Q12}$ and $\mathsf{Q13}$.
\end{proof}

\section{Pure-State Stabilization}
\label{sec:pure_stabilization}

We now derive the fixed-axis criterion used in Section~\ref{sec:control_consequences} and use it to decide pure-state stabilizability in every stratum. Stationarity and attraction enter separately. For a finite set of active jumps, the pure-stationary-state condition combines a common jump eigenstate with the corresponding Hamiltonian compatibility condition \cite[Thm.~1]{kraus2008preparation}. For the orbit averages considered here, the purity-variance identity restricts the candidate states, while the symmetric Bloch matrix determines their attraction.

\begin{definition}[Pure-state global asymptotic stability]
\label{def:pure_gas}
A generator $X\in W(G)$ has \emph{pure-state GAS} if its quantum dynamical
semigroup has a pure stationary state $\rho_\ast$ and every initial density
operator converges to $\rho_\ast$ in trace norm.
\end{definition}

For axial control, set $R_\theta=e^{\theta J_n}$ and $U_\theta=e^{-i\theta n\cdot\sigma/2}$. Write $L_\theta=U_\theta LU_\theta^\dagger$ and $Q_\theta=R_\theta QR_\theta^T$ for the rotated jump and its Bloch matrix.

Every $X\in W(G)$ has a representation $X=hJ_n+A_\mu$ with $\mu\geq0$, by the limiting argument in Lemma~\ref{lem:qubit_orbit_average}. In its approximating sums, the Hamiltonian parts converge as the antisymmetric matrix parts, while the trace bounds the measures\textquotesingle{} masses. Compactness of the control orbit then supplies a finite limiting measure.

\begin{lemma}[Pure candidates and uniform contraction for one jump]
\label{lem:qubit_pure_candidates}
Let $0<t<s$ in the principal form of an axially controlled system from
Corollary~\ref{cor:single_jump_principal_form}. If
$X=hJ_n+A_\mu$ has positive dissipative mass and fixes a pure state, that
state is a common right eigenstate of the rotated jumps for
$\mu$-almost every orbit point. The unrotated principal-form jump has the two right-eigenstate Bloch vectors
\[
r_\varepsilon=(\varepsilon\sqrt{1-q^2},0,q),
\qquad \varepsilon\in\{-1,1\},\qquad q=t/s.
\]
At orbit angle $\theta$, the candidates are $R_\theta r_\varepsilon$. Every rotated matrix part satisfies $Q_\theta\preceq-2t^2\mathbb I$.
\end{lemma}

\begin{proof}
Write a pure state as $\rho=|\psi\rangle\langle\psi|$. Stationarity and the
vanishing Hamiltonian contribution to the purity derivative give
\begin{equation}
\label{eq:pure_variance_identity}
0=\Tr\!\left[\rho X(\rho)\right]
=-\int\!\left(
\langle L_\theta^\dagger L_\theta\rangle_\rho
-|\langle L_\theta\rangle_\rho|^2
\right)d\mu(\theta).
\end{equation}
The integrand is nonnegative, so it vanishes almost everywhere. Thus
$L_\theta|\psi\rangle$ is proportional to $|\psi\rangle$ at every active
orbit point. The variance is a continuous function of the orbit parameter, so
it vanishes at every point of $\operatorname{supp}\mu$. In the principal frame,
\[
L=\begin{pmatrix}0&s+t\\ s-t&0\end{pmatrix},
\]
whose two eigenlines give the displayed Bloch vectors. The eigenvalues of the
symmetric Bloch matrix are $-2t^2$, $-2s^2$, and
$-2(s^2+t^2)$. Orthogonal conjugation preserves the stated upper bound.
\end{proof}

\begin{theorem}[Axial pure-state stabilization criterion]
\label{thm:axial_pure_gas}
Let $0<q<1$, let the control algebra be $\R J_n$ for a unit axis $n$, and write
\[
x=|n\!\cdot\!\hat\alpha|^2,\qquad
y=|n\!\cdot\!\hat\beta|^2,\qquad
\zeta=|n\!\cdot\!\hat\nu|^2.
\]
The wedge contains a pure-state GAS generator if and only if
\begin{equation}
\label{eq:pure_gas_discriminant}
q^2x=(1-q^2)\zeta.
\end{equation}
When this criterion holds, one dissipative orbit ray suffices if $y>0$,
whereas two rays are necessary and a half-turn pair suffices if $y=0$.
\end{theorem}

The proof exhausts the orbit supports compatible with a pure stationary state before constructing the stabilizers. A target off the control axis can receive support from one or both eigenstate branches, whereas an axial target requires a separate argument. The uniform contraction bound then turns the constructed stationary states into global attractors.

\begin{proof}
Put $c=\sqrt{1-q^2}$ and use the principal basis. At the candidate
$r_\varepsilon=(\varepsilon c,0,q)$, the dissipative Bloch force is
\begin{equation}
\label{eq:pure_candidate_force}
f_\varepsilon=Qr_\varepsilon+b
=2s^2qc\,(-\varepsilon q,0,c).
\end{equation}
A GAS generator must have positive dissipative mass. Otherwise its flow is Hamiltonian and preserves distances between Bloch vectors.

Suppose first that a pure stationary vector $r$ is not parallel to $n$.
Lemma~\ref{lem:qubit_pure_candidates} restricts the support of the orbit
measure to rotations for which $r=R_\theta r_\varepsilon$. Each eigenstate
branch contributes at most one such rotation. If only one branch has positive
mass, rotate the principal frame about $n$ to that support point. Write
$n_\alpha=n\cdot\hat\alpha$ and similarly for the other components. The
$xz$-components of stationarity force $n_\beta\neq0$. Solving
the stationarity equation after normalizing this positive mass to one gives
\begin{equation}
\label{eq:pure_gas_signed_condition}
\varepsilon c\,n_\nu=q n_\alpha,
\qquad
h=\frac{2\varepsilon t\sqrt{s^2-t^2}}{n_\beta},
\end{equation}

The squared relation in \eqref{eq:pure_gas_signed_condition} is
\eqref{eq:pure_gas_discriminant}.

It remains to consider support from both eigenstate branches. Their rotation
orbits can meet only if their heights along the control axis agree. Since
\[
n\cdot r_\varepsilon=\varepsilon c\,n_\alpha+q n_\nu,
\]
this forces $n_\alpha=0$. Equation~\eqref{eq:pure_candidate_force} then shows
that the force from either branch has the same axial component,
$n\cdot f_\varepsilon=2s^2qc^2n_\nu$. Positive orbit weights cannot cancel
this component, while $n\cdot(n\times r)=0$. Stationarity therefore gives
$n_\nu=0$, and \eqref{eq:pure_gas_discriminant} follows again.

If the stationary vector is parallel to $n$, every active eigenstate must
already be parallel to the axis. Hence $|n_\alpha|=c$, $n_\beta=0$, and
$|n_\nu|=q$, which also gives \eqref{eq:pure_gas_discriminant}. These cases
prove necessity.

For sufficiency, assume \eqref{eq:pure_gas_discriminant}. When $y>0$, choose
$\varepsilon$ so that the signed equality in
\eqref{eq:pure_gas_signed_condition} holds. The same equation supplies a
one-orbit-ray generator fixing $r_\varepsilon$. When $y=0$, the criterion and
$x+\zeta=1$ give $x=c^2$ and $\zeta=q^2$, so $n$ is parallel to one of the
vectors $r_\varepsilon$. Let $R_\pi=2nn^T-\mathbb I$. The force
$f_\varepsilon$ is tangent to the sphere at $r_\varepsilon$, and the half-turn
reverses it. Consequently,
\[
\frac12\left(\Dcal_L+\Ad_{R_\pi}\Dcal_L\right)
\]
fixes that pure state. A single orbit ray cannot do so because
$r_\varepsilon$ is parallel to $n$, so the axial Hamiltonian force
$n\times r_\varepsilon$ vanishes, whereas the single-ray dissipative force
\eqref{eq:pure_candidate_force} is nonzero for $0<q<1$. The earlier single-branch argument excludes off-axis targets when $y=0$, because it requires $n_\beta\neq0$. Together, these cases prove
minimality of the two-point support.

Finally, an orbit average of mass $\lambda(\mu)>0$ has symmetric matrix part
$Q_\mu\preceq-2t^2\lambda(\mu)\mathbb I$. Adding $hJ_n$ changes only the
antisymmetric part. The resulting affine matrix is Hurwitz, meaning that all its eigenvalues have strictly negative real parts, so its fixed point
is unique and every Bloch trajectory converges to it exponentially. The pure
fixed points constructed above are therefore GAS.
\end{proof}

\paragraph{Convergence rate.}
For the interior channels $0<t<s$ and positive dissipative mass $\lambda(\mu)$ in the proof, write $\delta r=r-r_\ast$ and use $\tau$ for time. The Hamiltonian matrix is antisymmetric, so
\[
\frac{d}{d\tau}\|\delta r\|_2^2
=2\,\delta r^TQ_\mu\delta r
\leq-4t^2\lambda(\mu)\|\delta r\|_2^2.
\]
Integration gives
\begin{equation}
\label{eq:pure_gas_convergence_rate}
\|\rho(\tau)-\rho_\ast\|_1
\leq e^{-2t^2\lambda(\mu)\tau}\|\rho(0)-\rho_\ast\|_1,
\end{equation}
because the trace norm of a qubit state difference equals the Euclidean distance of its Bloch vectors. The bound refers to the effective generator at its stated dissipative strength.

At $q=0$, every element of the wedge is unital and fixes the center of the
Bloch ball, so no pure state is GAS. At $q=1$, the unrotated extremal-damping
generator already has a unique globally attractive pure fixed point.

\begin{corollary}[Pure-state stabilization atlas]
\label{cor:qubit_pure_gas_atlas}
The capability of the thirteen strata is given in
Table~\ref{tab:qubit_pure_gas}.
\end{corollary}

\begin{table}[H]
\centering
\small
\begin{tabular}{c p{0.55\textwidth} p{0.18\textwidth}}
\hline
stratum & pure-state GAS in $W(G)$ & dissipative orbit rays \\
\hline
$\mathsf{Q1}$ & exactly at $q=1$ & one \\
$\mathsf{Q2},\mathsf{Q11}$ & no & --- \\
$\mathsf{Q3}$ & yes for extremal damping, no for dephasing & one when yes \\
$\mathsf{Q4}$ & no & --- \\
$\mathsf{Q5}$ & the $\hat\beta$ branch and the $q=1$ endpoint & one \\
$\mathsf{Q6}$ & yes & two \\
$\mathsf{Q7},\mathsf{Q8},\mathsf{Q12}$ & no & --- \\
$\mathsf{Q9}$ & all $q=1$ points. For $0<q<1$, exactly when \eqref{eq:pure_gas_discriminant} holds & one at $q=1$. For $0<q<1$, one if $y>0$ and two if $y=0$ \\
$\mathsf{Q10}$ & all $q=1$ points. For $0<q<1$, exactly when \eqref{eq:pure_gas_discriminant} holds & one \\
$\mathsf{Q13}$ & yes for every $q>0$ & one \\
\hline
\end{tabular}
\caption{Existence of a pure-state GAS generator in each classified wedge.
The last column gives the smallest support size in the control-conjugacy orbit
of the dissipator.}
\label{tab:qubit_pure_gas}
\end{table}

\begin{proof}
Rows $\mathsf{Q2}$ and $\mathsf{Q11}$ are Hamiltonian-only. The unital rows
$\mathsf{Q7}$, $\mathsf{Q8}$, and $\mathsf{Q12}$ are excluded by the $q=0$
endpoint argument. For $0<q<1$, the residual data displayed after
Table~\ref{tab:qubit_strata} reduce the remaining axial cases to
Theorem~\ref{thm:axial_pure_gas}. Stratum $\mathsf{Q4}$ has $n\parallel\hat\nu$.
Stratum $\mathsf{Q5}$ has the $\hat\alpha$ and $\hat\beta$ branches, and
stratum $\mathsf{Q6}$ has $q^2=x=\zeta=1/2$ with $y=0$. The criterion remains a
within-stratum condition in $\mathsf{Q9}$ and $\mathsf{Q10}$. On the magic
locus of $\mathsf{Q10}$, imposing $y=0$ together with
\eqref{eq:pure_gas_discriminant} gives
$q^2=x=\zeta=1/2$, which is the excluded $\mathsf{Q6}$ degeneration. Hence
every interior positive case in $\mathsf{Q10}$ has one-ray support. In
$\mathsf{Q1}$ there is no coherent control and
$W=\R_{\geq0}\Dcal_L$. For $0<q<1$ the two pure candidates have the nonzero
force \eqref{eq:pure_candidate_force}, so neither is stationary. The $q=0$ and $q=1$ arguments above settle all channel endpoints,
including the two cases in $\mathsf{Q3}$. For $\mathsf{Q13}$,
full control contains the $\hat\beta$ one-ray construction for every
$0<q<1$ and the bare damping generator at $q=1$.
\end{proof}

\begin{remark}[Full-control targets]
For $\mathsf{Q13}$, jointly rotating the dissipator, the Hamiltonian, and its
pure fixed point moves the one-ray construction to any point of the Bloch
sphere. If the original laboratory-frame dissipator is held fixed, a constant
generator using a Hamiltonian direction in the coherent-control edge stabilizes
one of its right eigenstates. That direction may be directly actuated or
synthesized in the closed-control sense. Arbitrary targets require the
orbit-conjugated wedge construction.
\end{remark}

\paragraph{Using the classification.}
Separate the jump-free specifications using Table~\ref{tab:qubit_pure_gas}. For a nonzero channel, reduce the jump to principal form and compute the channel discriminants of Section~\ref{sec:channel_normal_form}. Together with the declared coherent controls, these data assign the Levi-data stratum through Table~\ref{tab:qubit_strata}. The stabilization atlas in Table~\ref{tab:qubit_pure_gas} covers the uncontrolled and full-control cases as well. On an interior axial branch, $0<q<1$, evaluate \eqref{eq:control_task_discriminant}. A nonzero value excludes pure-state GAS. If $\Delta=0$ and $y>0$, use one orbit ray with the compensating Hamiltonian in \eqref{eq:pure_gas_signed_condition}. If $\Delta=0$ and $y=0$, use the half-turn pair. At the channel endpoints, $q=0$ admits no pure-state GAS, while $q=1$ requires only the unrotated damping ray.

\section{Structural Interpretation of the Thirteen Strata}
\label{sec:physical_consequences}

Rows \(\mathsf{Q1}\)--\(\mathsf{Q3}\) have solvable type, \(\mathsf{Q11}\) has semisimple type, \(\mathsf{Q4}\), \(\mathsf{Q7}\), \(\mathsf{Q8}\), and \(\mathsf{Q12}\) have mixed reductive type, and \(\mathsf{Q5}\), \(\mathsf{Q6}\), \(\mathsf{Q9}\), \(\mathsf{Q10}\), and \(\mathsf{Q13}\) have mixed non-reductive type. Thus this minimal nontrivial family already realizes all four Levi types.

\paragraph{Solvable and semisimple types.}
For \(\mathsf{Q1}\)--\(\mathsf{Q3}\), Lie's theorem supplies an invariant flag after complexification. This is the representation-theoretic content of solvable type. Exact solvability, conserved observables, and mixing laws require additional dynamical data. At the opposite extreme, Corollary~\ref{cor:hamiltonian_rigidity} implies that a physical wedge with semisimple system algebra is Hamiltonian-only. Accordingly, the sole semisimple stratum \(\mathsf{Q11}\) contains no jump channel and no dissipative pure-state stabilizer.

\paragraph{Mixed types.}
In mixed reductive type the Levi action on the radical is trivial, but the wedge need not split into its semisimple and radical slices. Stratum \(\mathsf{Q4}\) is an explicit example: its four-dimensional wedge is not recovered from the two-dimensional span of its edge line and channel-average ray. In mixed non-reductive type, coupling can arise from a nontrivial Levi action, a nonabelian radical, or both. The Fiber Compatibility Condition in Theorem~\ref{thm:levi_type_main} records how edge motion constrains the corresponding fibers.

\paragraph{Relation to symmetry and spectral classifications.}

The Levi classification and existing symmetry or spectral classifications answer different questions. Symmetry and complex-spacing analyses attach invariants to a selected Lindbladian or ensemble \cite{sa2020complex,sa2023symmetry,kawabata2023symmetry}. Here the input is a generator-labelled control family, and the output describes its generated algebra together with its minimal local wedge. A single Lindbladian viewed as a singleton family generates a one-dimensional solvable algebra, whereas adjoining coherent controls can place the same dissipative direction in a mixed stratum.

\paragraph{What the structure predicts.}
A stratum label can exclude pure-state stabilization or reduce it to a residual condition. The $\mathsf{Q9}$ pair in Example~\ref{ex:q9_control_contrast} lies in the same generated algebra and the same stratum, yet its control-axis modulus changes the answer. Once existence is settled, orbit support distinguishes a single rotated dissipator from an average of several orientations, as Example~\ref{ex:q6_half_turn} demonstrates.

\paragraph{Protected substructures require representation data.}
Determining whether a protected subsystem exists requires the represented noise operators and full generator blocks in addition to the Levi tuple. Pure-state stabilization asks whether some generator in the wedge has a target property. A common decoherence-free subspace or noiseless subsystem instead requires compatibility with every specified generator. A noiseless subsystem is a protected tensor factor within a subspace. A one-dimensional cofactor gives a decoherence-free subspace. A reducing subspace is invariant under both an operator algebra and its adjoints, a dagger-closed algebra is closed under Hermitian adjoints, and its commutant consists of all operators commuting with it. In this setting, the noise-algebra decomposition determines noiseless multiplicities from the represented algebra and its commutant \cite{knill2000theory}. An initialized noiseless subsystem requires unitary evolution of the reduced protected state for every initial state supported in the tensor block, independently of the cofactor state. The necessary-and-sufficient Markovian block conditions of Ticozzi--Viola also allow off-diagonal noise and Hamiltonian terms that compensate each other \cite[Defs.~7--8 and Cor.~3]{ticozzi2008quantum}.

\section{Conclusion}
\label{sec:conclusion}

The drift-free, single-channel qubit family has thirteen Levi-data strata, assigned by the control algebra and channel discriminants. Every classified wedge is global, its reversible edge is the coherent-control algebra, and four strata are Lie semialgebras.

The stabilization atlas turns this structure into a control test. Five strata contain both positive and negative cases, separated by channel moduli. For interior channels with axial control, the scalar criterion resolves this dependence. The $\mathsf{Q9}$ contrast shows that even the generated Lie algebra and stratum together can leave stabilizability undecided. In $\mathsf{Q6}$, stabilization requires two control-conjugate dissipative orientations. Multiple channels introduce relative-orientation data, while one-sided drift changes the algebra or wedge. Those extensions and protected-subsystem questions require further invariants.

\section*{Acknowledgments}
OpenAI Codex (GPT-6) was used primarily for manuscript review. It also assisted with the preparation and checking of auxiliary verification scripts and with language and formatting revisions. The author directed this use, verified the adopted outputs, and takes responsibility for the content of the article.

\section*{Data Availability}
The results are analytical and do not rely on experimental or observational data. The auxiliary scripts used for symbolic and numerical checks are available from the author upon reasonable request.

\appendix

\section{Wedge Construction and Dissipation Identities}
\label{app:wedge_construction}

This appendix supplies the finite construction used in Definition~\ref{def:DLWP} and the dissipation identities used to interpret the classification. The construction applies to a general finite specification, while the qubit classification exploits the more explicit orbit-average formula of Lemma~\ref{lem:qubit_orbit_average}.

\subsection{Construction of the minimal wedge}
\begin{proposition}[The iterative closed-wedge construction yields $W(G)$]
\label{prop:construction_W(G)_saturation}
Set $W_0=\overline{\cone}(G)$. Given $W_k$, let
$E_k=W_k\cap(-W_k)$, let $H_k=\langle\exp(E_k)\rangle$, and define
\[
W_{k+1}=\overline{\cone}\!\left(\Ad_{H_k}W_k\right).
\]
This sequence reaches a fixed point after finitely many edge enlargements, and
that fixed point is the smallest closed Lie wedge in $\g(G)$ containing $G$,
namely $W(G)$. In particular, $W(G)\subseteq L(\Sigma(G))$.
\end{proposition}

\begin{proof}
Closed Lie wedges containing $G$ form a nonempty family, since $\g(G)$ is one,
and are stable under arbitrary intersections. Hence the smallest such wedge
exists. If $V$ is any closed Lie wedge containing $G$, induction gives
$W_k\subseteq V$: from $E_k\subseteq E(V)$, edge invariance of $V$ absorbs the
$H_k$-orbit and its closed conical hull. The subspaces $E_k$ are nondecreasing,
so finite dimensionality gives $E_{k+1}=E_k$ at some stage. Then $W_{k+1}$ is
invariant under the group generated by its own edge and is therefore a closed
Lie wedge. It lies in every $V$, so $W_{k+1}=W(G)$, and the next step is
unchanged. This orbit-and-edge iteration adapts the template of
\cite[Sec.~II-D]{OMeara2011IllustratingTG}. The inclusion in the tangent wedge
is \eqref{eq:saturate_vs_tangent}.
\end{proof}
\Needspace{8\baselineskip}

\paragraph{Dissipation Functional and Hamiltonian Rigidity.}
The two results below are recalled from Ref.~\cite[Sec.~2.4.3, Lemma~2.16 and Cor.~2.17]{wangthakkar2026geometry}. The Kossakowski trace defines a linear functional whose vanishing on the physical wedge identifies the Hamiltonian edge and constrains the commutator ideal.
\begin{lemma}[Dissipation functional]
\label{lem:dissipation_functional}
Let $\{F_j\}_{j=1}^{d^2-1}$ be a basis of traceless Hermitian operators on $\Hcal$ with Gram matrix
$G_{jk} = \Tr(F_kF_j)$, and let $\Lcal \in \g_{LK}$ have GKLS data $(H, K)$ in that
basis. Then, as a superoperator on $\Bcal(\Hcal)$,
\begin{equation}
\label{eq:dissipation_functional}
\Tr \Lcal \;=\; -\,d\,\Tr(K G),
\end{equation}
independently of $H$. In the orthonormal normalization $G = \mathbb{I}$ this reads
$\Tr\Lcal = -d\,\Tr K$. In the normalization $\Tr(F_iF_j) = 2\delta_{ij}$ used in
the unnormalized Gell--Mann convention, one has
$G = 2\cdot\mathbb{I}$ and $\Tr\Lcal = -2d\,\Tr K$.
\end{lemma}
Write $\varphi(\Lcal) := \Tr K$ for the orthonormal normalization, so that
$\varphi = -\tfrac{1}{d}\Tr(\cdot)$ is linear on $\g_{LK}$. On the GKLS wedge $K \succeq 0$,
so $\varphi \ge 0$ on $\WGKLS$, and $\varphi(\Lcal) = 0$ forces $K = 0$: the functional
vanishes on $\WGKLS$ exactly at the Hamiltonian generators, which by
Lemma~\ref{lem:gkls_edge} are the edge $E(\WGKLS) \cong \su(d)$. It measures the total
dissipation strength.

\begin{corollary}[Hamiltonian rigidity of the commutator ideal]
\label{cor:hamiltonian_rigidity}
Let $\g \subseteq \g_{LK}$ be a subalgebra and $W \subseteq \WGKLS \cap \g$ a wedge. Then
\begin{equation}
\label{eq:hamiltonian_rigidity}
W \cap [\g,\g] \;\subseteq\; E(\WGKLS) \cap \g ,
\end{equation}
that is, every element of $W$ lying in the commutator ideal is purely Hamiltonian.
Consequently, with $\g = \s \ltimes \rr$ the Levi decomposition:
\begin{enumerate}
  \item $W \cap \s \subseteq E(\WGKLS)$: the wedge meets the Levi factor only in
        Hamiltonian directions.
  \item if $\g$ is semisimple then $W$ consists entirely of Hamiltonian generators, so a
        semisimple system algebra synthesizes no dissipation at all.
  \item $\varphi(X) = \varphi(\pi_{\rr}(X))$ for every $X \in \g$, where $\pi_{\rr}$ is the
        projection onto $\rr$ along the Levi section: all dissipation strength is carried by
        the radical component.
  \item for any $\Lcal_1, \Lcal_2 \in \g_{LK}$, the bracket $[\Lcal_1,\Lcal_2]$ lies in
        $\WGKLS$ if and only if it is purely Hamiltonian.
\end{enumerate}
\end{corollary}

\section{Normalization and Structural Proofs for Qubits}
\label{app:qubit_structural_proofs}

\subsection{Normalization, gauge, and orthogonal symmetries}
\label{app:normalization_proof}

The first proof justifies the equivalences used to identify control specifications. The following propositions explain the role of the traceless-jump gauge and delimit the allowed linear changes of Bloch coordinates.

\paragraph{Proof of Lemma~\ref{lem:qubit_normalization}.}
\begin{proof}
(i) Conjugation by $u$ is an automorphism of $\gLK$ that maps $\WGKLS$ onto itself, carries $\cone(G)$ to $\cone(u \cdot G)$, and commutes with each step of the construction in Proposition~\ref{prop:construction_W(G)_saturation}, because edges, closed conic hulls, and adjoint orbits are all defined from the cone and bracket alone. It therefore maps $(\g(G), W(G))$ isomorphically onto $(\g(u\cdot G), W(u \cdot G))$, matching radicals, Levi factors, edges, radical wedges, and the translation ideal. Every entry of $T$ is a dimension or an isomorphism class of one of these objects.

(ii) A direct calculation gives the jump transformation. For $\Theta \Dcal_L \Theta$ one finds $(L\rho^T L^\dagger)^T = \bar L \rho \bar L^\dagger$ and likewise for the anticommutator, so the conjugated dissipator is $\Dcal_{\bar L}$, which is again traceless. For the coherent part, $(-i[H, \rho^T])^T = -i[-\bar H, \rho]$. Complete positivity is preserved because the Kossakowski matrix transforms as $K \mapsto S\bar{K}S$ with $S = \mathrm{diag}(1,-1,1)$, an entrywise conjugation composed with a real congruence, both of which preserve positive semidefiniteness. Since $\sigma_y$ is the only antisymmetric Pauli matrix, the induced Bloch map is $R_y$. Conjugation by $R_y$ is again an automorphism of the ambient pair, and the argument of (i) applies verbatim. The extension is genuine, as $\det R_y = -1$.

(iii) Let two specifications have bidirectional coherent lists generating the same algebra $\mf e_0$ and dissipative generators on the same positive ray. Any closed Lie wedge containing either coherent list contains that list in its edge. Since the edge is a Lie algebra, it contains $\mf e_0$. Conversely, a closed Lie wedge whose edge contains $\mf e_0$ contains either coherent list. Thus the closed Lie wedges containing either specification are exactly those containing $\mf e_0$ and the shared channel ray. Minimality gives the same $W(G)$, and Lie generation gives the same $\g(G)$. Hence $T(G)$ is unchanged.
\end{proof}

\begin{proposition}[Gauge sensitivity of separated specifications]
\label{prop:qubit_gauge_sensitivity}
Let $L$ be a nonzero traceless jump and $c\in\C$, and put $H_c=\tfrac{i}{2}(\bar cL-cL^\dagger)$. Then
\[
\Dcal_{L+c\mathbb I}=\Dcal_L+\Lcal_{H_c},\qquad
\Lcal_{H-H_c}+\Dcal_{L+c\mathbb I}=\Lcal_H+\Dcal_L
\]
for every Hermitian $H$. If $H_c\neq0$, the positive rays of $\Dcal_{L+c\mathbb I}$ and $\Dcal_L$ differ.
\end{proposition}

For $H_c\neq0$, the shifted dissipator has a coherent component in the traceless-jump gauge and is not a channel generator of Definition~\ref{def:qubit_family}. The following proof gives an explicit damping witness whose separated-generator closure has dimension $2$ at $c=0$ and $12$ at $c=3/5$.

\paragraph{Separated-generator gauge witness.}
\label{app:gauge_witness}

\begin{proof}[Proof of Proposition~\ref{prop:qubit_gauge_sensitivity}]
Expanding the shifted dissipator gives $\Dcal_{L+c\mathbb I}-\Dcal_L=\tfrac12[\bar cL-cL^\dagger,\cdot]=\Lcal_{H_c}$, proving the compensated identity. Its nonzero coherent component changes the channel ray when $H_c\neq0$, since the symmetric Bloch matrix remains unchanged and nonzero. Treating the compensated sum as a single autonomous generator therefore does not change the Lie closure it generates. The dimension jump concerns the different separated list $\{Z,Y_c\}$ below, in which the compensating term is not recombined with $Y_c$. In affine Bloch coordinates write $e_z=(0,0,1)^T$ and put
\[
D=\bigl(\operatorname{diag}(-\tfrac12,-\tfrac12,-1),-e_z\bigr),\qquad
Z=(J_z,0),\qquad Y_c=D+c(J_y,0).
\]
For $c=3/5$, define $g_1=Z$, $g_2=Y_c$, and
\[
\begin{aligned}
g_3&=[g_1,g_2],& g_4&=[g_3,g_1],& g_5&=[g_3,g_2],&
g_6&=[g_4,g_2],\\
g_7&=[g_4,g_5],& g_8&=[g_5,g_3],& g_9&=[g_6,g_2],&
g_{10}&=[g_6,g_7],\\
g_{11}&=[g_6,g_8],& g_{12}&=[g_9,g_{11}].
\end{aligned}
\]
Flattening the twelve affine $3\times4$ blocks row by row gives a matrix with determinant
$-3^{20}/(2^4 5^{20})$, so the separated list generates all of
$\gl(3,\R)\ltimes\R^3$. At $c=0$, $[Z,D]=0$, and the Lie algebra generated by $Z$ and $D$ remains two-dimensional. These calculations establish the $2\to12$ change for separated generators.
\end{proof}

\paragraph{Orthogonal Bloch symmetries.}
For the next witness, set $\sigma_-=(\sigma_x-i\sigma_y)/2=|1\rangle\langle0|$ in the computational basis. For a linear map $X$ on qubit operators, use the unnormalized Choi matrix
\[
C_X=\sum_{i,j=0}^{1}|i\rangle\langle j|\otimes X(|i\rangle\langle j|),
\quad |\Omega\rangle=|00\rangle+|11\rangle,
\quad P_\Omega=I-\tfrac12|\Omega\rangle\langle\Omega|.
\]
Conditional complete positivity means that $P_\Omega C_XP_\Omega$ restricted to $\Omega^\perp$ is positive semidefinite.

\begin{proposition}[Maximality of orthogonal Bloch symmetries]
\label{prop:qubit_bloch_symmetry_maximality}
A Bloch conjugation $(Q,b)\mapsto(TQT^{-1},\allowbreak Tb)$ with $T \in \mathrm{GL}(3,\R)$ maps $\WGKLS$ onto itself if and only if $T$ is orthogonal. In particular, $T=2\mathbb I$ maps the extremal damping generator $\Dcal_{\sigma_-}$ to a map that violates conditional complete positivity.
\end{proposition}

\begin{proof}
Let $T \in \mathrm{GL}(3,\R) \setminus \mathrm{O}(3)$ and polar-decompose $T = RP$ with $R \in \mathrm{O}(3)$ and $P$ positive definite, $P \neq \mathbb{I}$. Symmetries form a group containing $\mathrm{O}(3)$, so $T$ can be a symmetry only if $P$ is, and after passing to $P^{-1}$ if necessary $P$ has an eigenvalue $\lambda > 1$. Conjugate the extremal damping generator with axis along the corresponding eigenvector: axial symmetry makes its matrix part commute with $P$, the translation along the axis is scaled by $\lambda$, and the fixed point of the flow moves outside the state space, so the conjugated generator violates conditional complete positivity. For $P = 2\,\mathbb{I}$, the conditional Choi compression in the stated unnormalized convention has eigenvalue $-1/2$. The anisotropic choice $P = \mathrm{diag}(2,1,1)$ gives a second witness.
\end{proof}

Parts \textup{(i)} and \textup{(ii)} of Lemma~\ref{lem:qubit_normalization} supply the allowed orthogonal reorientations, while part \textup{(iii)} permits equivalent control lists and positive channel scaling. Corollary~\ref{cor:single_jump_principal_form} gives the principal channel form. Propositions~\ref{prop:qubit_gauge_sensitivity} and~\ref{prop:qubit_bloch_symmetry_maximality} delimit the fixed gauge and the ambient linear symmetries.

\subsection{Edge rigidity}
\label{app:edge_proof}

The nonnegative dissipation functional isolates the Hamiltonian slice at every stage of the wedge construction. This is the step that prevents dissipation from creating additional reversible directions. The final part of the proof treats one-sided Hamiltonian drifts as well as the drift-free family.

\paragraph{Proof of Lemma~\ref{lem:qubit_edge_rigidity}.}
\begin{proof}
(i) The wedge satisfies $W(G) \subseteq \g(G) \cap \WGKLS$, so by Lemma~\ref{lem:gkls_edge} its edge lies in $E(\WGKLS) \cong \su(2)$. Edges of Lie wedges are subalgebras, and $\so(3)$ has no two-dimensional subalgebra, since the bracket of two independent rotation generators points along the cross product of their axes.

(ii) Consider the linear functional $\varphi(X) := \operatorname{tr} K(X)$, the trace of the Kossakowski matrix in a fixed basis. It is nonnegative on $\WGKLS$, vanishes exactly on the Hamiltonian generators, and is invariant under unitary conjugation, which acts on $K$ by congruence with a unitary. The zero set $\{\varphi = 0\} \cap \WGKLS = E(\WGKLS)$ is thus an exposed face fixed by the only adjoint actions that the construction in Proposition~\ref{prop:construction_W(G)_saturation} invokes. Since closed conic hulls of closed conic hulls collapse, the stage-$k$ wedge is $W_k(G) = \overline{\operatorname{cone}}\, \mathrm{Ad}_{H_{k-1}}(G)$, where $H_{k-1}$ is the group generated by the edge subgroups of the earlier stages. 

We claim, by simultaneous induction, that the towers for $G$ and for $G_H$ have equal Hamiltonian slices and equal edges at every stage. Suppose they share their conjugating groups up to stage $k$. An element of $W_k(G) \cap \ker\varphi$ is a limit of combinations $Y^{(n)} + \sum_i \Lambda_i^{(n)}$, with $Y^{(n)}$ a conic combination of rotated Hamiltonian generators and $\Lambda_i^{(n)}$ one of rotated copies of the $i$th dissipative generator with total mass $\lambda_i^{(n)}$. Here additivity enters: $\varphi$ is constant on each orbit, so $\varphi$ of the combination equals $\sum_i \lambda_i^{(n)} \operatorname{tr} K_i \to 0$ with every $\operatorname{tr} K_i > 0$, forcing each $\lambda_i^{(n)} \to 0$. Since the specification is finite and the conjugating groups are compact, each $\Lambda_i^{(n)} \to 0$ and the limit lies in $W_k(G_H) \cap \ker\varphi$. The reverse inclusion is plain. 

Both edges lie in their slices, $\varphi$ vanishing on linear subspaces of the wedge, and conversely the lineality of a slice lies in the edge, so equal slices give equal edges, hence the same conjugating group at stage $k+1$, closing the induction. The final edges then agree.

(iii) By (ii) we may assume $G = G_H$. If the controls span at least a two-dimensional set of axes, or together with the drift lineality do, then $\mf{e} = \so(3)$ and the wedge contains both signs of a generating set of $\so(3)$. The minimal closed Lie wedge containing a symmetric set contains its generated subalgebra, which gives $E = \so(3)$. 

Suppose $\mf{e} = \R J_n$ and some drift $J_m$ has $m \nparallel n$. Decompose $J_m = \cos\theta\, J_n + \sin\theta\, J_{m'}$ with $m' \perp n$. Adding the edge element $-\cos\theta\,J_n \in W$ to $J_m$ shows $J_{m'} \in W$, and conjugating by the half-turn $e^{\pi J_n}$ sends $J_{m'}$ to $-J_{m'}$, so $J_{m'}$ joins the edge. Together with $J_n$ it generates $\so(3)$. There is no further growth, the ambient being three-dimensional. 

Now suppose instead that every drift axis is parallel to $n$. The set $C := \{X \in \WGKLS : Q(X) \in \R J_n\}$ is a closed convex cone, contains $G_H$, has edge $\R J_n$, and is invariant under conjugation by $e^{tJ_n}$, which fixes the condition defining it. Hence $C$ is a closed Lie wedge containing $G_H$, minimality gives $W(G_H) \subseteq C$, and the edge cannot exceed $\R J_n$. 

The case $\mf{e} = 0$ is the same argument with $C = \overline{\cone}(G_H)$, whose lineality is trivial by hypothesis. Finally, for $G \in \mathcal{F}$ there are no drifts, so $\mf{e} = \mf{e}_0$ and only the no-growth branches occur.
\end{proof}

\subsection{Possible Levi factors}
\label{app:levi_menu_proof}

\paragraph{Proof of the semisimple-type menu.}
\begin{proof}
The Levi factor $\s$ of $\g(G) \subseteq \gl(3,\R) \levi \R^3$ is a semisimple subalgebra of the ambient algebra. Its intersection with the ambient radical $\R\mathbb{I} \oplus \R^3$ is a solvable ideal of $\s$, hence zero, so the quotient map onto $(\gl(3,\R) \levi \R^3)/(\R\mathbb{I} \oplus \R^3) \cong \sll(3,\R)$ is injective on $\s$. Complexifying, $\s_\C$ is a semisimple subalgebra of $\sll(3,\C)$ acting faithfully on $\C^3$. A semisimple algebra with more than one simple ideal has no faithful module of dimension three: an irreducible module of a direct sum is an outer tensor product, so a constituent of dimension at most three admits a nontrivial action of at most one simple ideal, and faithfulness would then need two disjoint constituents of dimension at least two, exceeding three. Thus $\s_\C$ is simple with a faithful three-dimensional module, which forces $\s_\C \in \{\sll(2,\C), \sll(3,\C)\}$, because among the simple complex algebras only these admit a faithful module of dimension at most three, the next smallest cases being $\mf{sp}(4,\C)$ with four and $G_2$ with seven. The real forms of $\sll(2,\C)$ are $\sll(2,\R)$ and $\su(2) \cong \so(3)$, and both occur inside $\sll(3,\R)$. If $\s_\C = \sll(3,\C)$ then $\dim_\R \s = 8$, the injection above is onto, and $\s \cong \sll(3,\R)$.
\end{proof}

\section{Levi Reconstruction}
\label{app:levi_gauge_details}

\subsection{Properties and reconstruction of the Levi data}

This appendix recalls the Levi decomposition and reconstruction framework of Ref.~\cite[Sec.~4]{wangthakkar2026geometry}. Fix the DLWP $(\g,W)$ and Levi splitting $\g=\s\ltimes\rr$ of Section~\ref{sec:levi_type_theorem}, with the components of Definition~\ref{def:fiber_cone}.

Theorem~\ref{thm:levi_type_main} gives a fiber-coordinate re-encoding. Under the linear isomorphism $(S,R)\mapsto S+R$, the graph of the fiber map is the wedge itself, so $(\Cquot,S\mapsto W(S))$ and $W$ determine one another.

\begin{lemma}[Structure of the Radical Wedge]
\label{lem:radical_wedge}
The radical wedge $W_{\rr} = W \cap \mf{r}$ is a Lie wedge in the Lie algebra $\mf{r}$. Its edge is $E(W_{\rr}) = E(W) \cap \mf{r}$.
\end{lemma}
\begin{lemma}[Structure of the Quotient Cone]
\label{lem:quotient_cone}
The quotient cone $C_{\text{quot}} = \pi_{\s}(W)$ is a convex cone in $\mf{s}$ (it need not be closed). The projected edge $E_{\text{quot}} = \pi_{\s}(E(W))$ is a subalgebra of $\mf{s}$ and is automatically closed. The quotient cone is invariant under the adjoint action of $E_{\text{quot}}$:
\begin{equation}
e^{\ad(S_0)}(C_{\text{quot}}) \subseteq C_{\text{quot}} \qquad \text{for all } S_0 \in E_{\text{quot}}.
\end{equation}
\end{lemma}
\begin{lemma}[Properties of the Fiber Map]
\label{lem:fiber_map_properties}
The Fiber Map $S \mapsto W(S)$ satisfies the following properties:
\begin{enumerate}
    \item $W(S)$ is a nonempty closed convex set in $\mf{r}$ for all $S \in C_{\text{quot}}$.
    \item $W(0) = W_{\rr}$. (The fiber above the origin is the Radical Wedge).
    \item Subadditivity: $W(S_1) + W(S_2) \subseteq W(S_1+S_2)$.
    \item Homogeneity (Scaling): $\lambda W(S) \subseteq W(\lambda S)$ for $\lambda \ge 0$. Since $W$ is a cone, equality holds for $\lambda > 0$.
    \item Closed graph: the graph $\Gamma = \{(S, R) \in \mf{s} \times \mf{r} \mid S \in C_{\text{quot}},\ R \in W(S)\}$ is closed in $\mf{s} \times \mf{r}$.
\end{enumerate}
\end{lemma}
\begin{proposition}[Necessity of the Fiber Compatibility Condition]
\label{lem:fiber_compatibility_main}
Let $W$ be a Lie wedge in $\g = \mf{s} \ltimes \mf{r}$. Then for every $X_0 = S_0+R_0 \in E(W)$ and every $S \in C_{\text{quot}}$,
\begin{equation}
\label{eq:FCC}
e^{\ad(X_0)}(S+W(S)) \subseteq e^{\ad(S_0)}(S) + W(e^{\ad(S_0)}(S)).
\end{equation}
\end{proposition}
We call \eqref{eq:FCC} the Fiber Compatibility Condition (FCC). The fiber on the right is well defined because $e^{\ad(S_0)}(S)\in C_{\text{quot}}$ by Lemma~\ref{lem:quotient_cone}.
\begin{definition}[Admissible fiber data]
\label{def:admissible_fiber_data}
Let $\g = \mf{s} \ltimes \mf{r}$. Suppose we are given a convex cone $C \subseteq \mf{s}$ containing $0$, not assumed closed, together with an assignment of a set $W(S) \subseteq \mf{r}$ to each $S \in C$, subject to the following conditions:
\begin{enumerate}
    \item[(D1)] each $W(S)$ is a nonempty closed convex subset of $\mf{r}$.
    \item[(D2)] (subadditivity) $W(S_1)+W(S_2) \subseteq W(S_1+S_2)$ for all $S_1, S_2 \in C$.
    \item[(D3)] (homogeneity) $\lambda W(S) \subseteq W(\lambda S)$ for all $\lambda \ge 0$ and $S \in C$.
    \item[(D4)] (closed graph) the set $\Gamma = \{(S,R) \mid S \in C,\ R \in W(S)\}$ is closed in $\mf{s} \times \mf{r}$.
\end{enumerate}
Define from these data the candidate edge
\begin{equation}
\label{eq:reconstructed_edge}
\Ecal \;=\; \bigcup_{S \in C \cap (-C)} \bigl(S + (W(S) \cap -W(-S))\bigr),
\end{equation}
and assume in addition:
\begin{enumerate}
    \item[(D5)] (base invariance) $e^{\ad(\pi_{\s}(X_0))}(C) \subseteq C$ for every $X_0 \in \Ecal$.
    \item[(D6)] (compatibility) every $X_0 \in \Ecal$ and every $S \in C$ satisfy the FCC \eqref{eq:FCC}, with $W(\cdot)$ the given assignment.
\end{enumerate}
We call $(C,S\mapsto W(S))$ satisfying \textup{(D1)}--\textup{(D6)} \emph{admissible fiber data}.
\end{definition}
\begin{proposition}[Reconstruction of the Wedge]
\label{lem:sufficiency}
Admissible fiber data in the sense of Definition~\ref{def:admissible_fiber_data} determine the Lie wedge
\begin{equation}
W \;=\; \bigcup_{S \in C} (S + W(S))
\end{equation}
in $\g$, with
\begin{equation}
W \cap \mf{r} = W(0), \qquad \pi_{\s}(W) = C, \qquad E(W) = \Ecal,
\end{equation}
and the fiber map of $W$ in the sense of Definition~\ref{def:fiber_cone} is the given assignment $S \mapsto W(S)$. No other subset of $\g$ has these components.
\end{proposition}
\begin{theorem}[Levi-Type Decomposition Theorem]
\label{thm:levi_type_main}
Let $(\g, W)$ be a Dynamical Lie Wedge Pair and let $\g = \mf{s} \ltimes \mf{r}$ be the Levi decomposition.
\begin{enumerate}
    \item \textbf{(Decomposition)} The components of Definition~\ref{def:fiber_cone} have the following properties. The radical wedge $W_{\rr}$ is a Lie wedge in $\mf{r}$ with edge $E(W_{\rr}) = E(W) \cap \mf{r}$ (Lemma~\ref{lem:radical_wedge}). The quotient cone $C_{\text{quot}}$ is a convex cone in $\mf{s}$, in general not closed, invariant under the adjoint action of the closed subalgebra $E_{\text{quot}} = \pi_{\s}(E(W))$ (Lemma~\ref{lem:quotient_cone}). The fiber map has nonempty closed convex fibers, satisfies $W(0) = W_{\rr}$, is subadditive and positively homogeneous, and has closed graph in $\mf{s} \times \mf{r}$ (Lemma~\ref{lem:fiber_map_properties}). The Fiber Compatibility Condition \eqref{eq:FCC} holds (Proposition~\ref{lem:fiber_compatibility_main}).
    \item \textbf{(Reconstruction)} Conversely, every admissible datum of Definition~\ref{def:admissible_fiber_data} defines, by Proposition~\ref{lem:sufficiency}, the Lie wedge $W = \bigcup_{S \in C}(S+W(S))$, the unique subset of $\g$ with the given components.
    \item \textbf{(Correspondence)} The two constructions are mutually inverse: decomposing a Lie wedge and reconstructing from the resulting data returns the same wedge, and reconstructing from admissible data and decomposing returns the same data.
\end{enumerate}
\end{theorem}
\begin{remark}[Typing and dependence on the Levi section]
Parts 2 and 3 give a canonical correspondence for arbitrary Lie wedges relative to a fixed Levi splitting $\g = \mf{s} \ltimes \mf{r}$. For the DLWP of a physical system, the wedge is $W(G)$, with data $(C_{\text{quot}}, S \mapsto W(S))$. The section is unique up to conjugation by $e^{\ad z}$ with $z$ in the nilradical, the largest nilpotent ideal (Malcev), and such a conjugation transports decomposition data to decomposition data.
\end{remark}

\Needspace{12\baselineskip}
\section{Algebraic Proof of the Qubit Classification}
\label{app:qubit_classification_proofs}

The distinction between the axial branches depends on whether brackets create rotations transverse to the control axis. The first two results settle this question for the matrix parts.

\begin{lemma}[Cross-weight regeneration]\label{lem:cross_weight_regeneration}
Let $\mathcal R_1=\Span\{J_x,J_y\}$ be the weight-one rotation sector, and use the symmetric sectors $\mathcal S_1$, $\mathcal S_2$, and $\mathcal S_0$ of Section~\ref{sec:channel_normal_form}. Then
\[
[\mathcal S_1,\mathcal S_2\oplus\mathcal S_0]\subseteq\mathcal R_1.
\]
For every nonzero $Z\in\mathcal S_2\oplus\mathcal S_0$, some $X\in\mathcal S_1$ has $[X,Z]\neq0$. This rotation generates all of $\mathcal R_1$ under the action of $J_z$.
\end{lemma}
\begin{proof}
The brackets are explicit on the sector bases:
\begin{equation}
[E_{13}+E_{31},\, \operatorname{diag}(1,1,-2)] = -3 J_y, \qquad [E_{23}+E_{32},\, \operatorname{diag}(1,1,-2)] = 3 J_x,
\end{equation}
together with $[E_{13}+E_{31},\, E_{11}-E_{22}] = -J_y$, $[E_{13}+E_{31},\, E_{12}+E_{21}] = J_x$, $[E_{23}+E_{32},\, E_{11}-E_{22}] = -J_x$, and $[E_{23}+E_{32},\, E_{12}+E_{21}] = -J_y$. For $Z = p(E_{11}-E_{22}) + \eta(E_{12}+E_{21}) + r\operatorname{diag}(1,1,-2)$ the pairing $X \mapsto [X, Z]$ sends $E_{13}+E_{31} \mapsto \eta J_x - (p+3r) J_y$ and $E_{23}+E_{32} \mapsto (-p+3r) J_x - \eta J_y$, a linear map of $\mathcal S_1$ into $\mathcal R_1$ with determinant $9r^2 - p^2 - \eta^2$. The map vanishes identically only at $Z = 0$, so every nonzero $Z$ produces a nonzero rotation, and it is onto except on the cone $9r^2 = p^2 + \eta^2$, where it has rank one. One nonzero rotation already suffices, since $[J_z, J_x] = J_y$ and $[J_z, J_y] = -J_x$ let $J_z$ rotate it through all of $\mathcal R_1$.
\end{proof}

\begin{corollary}[Axial closure dichotomy]
\label{cor:axial_closure_dichotomy}
Let a traceless symmetric matrix have a nonzero weight-one component. If its weight-two or deviatoric component is nonzero, its Lie closure with $J_z$ contains $\so(3)$ and then $\sll(3,\R)$. If both components vanish, the closure is the three-dimensional algebra $\sll(2,\R)$ and contains no additional rotation, weight-two, or deviatoric component.
\end{corollary}

\begin{proof}
The axial adjoint action separates the weight sectors and spans each nonzero rotating plane. The within-weight brackets $[\mathcal S_2,\mathcal S_2]$ and $[\mathcal S_1,\mathcal S_1]$ lie in $\R J_z$ and produce no transverse rotation. Lemma~\ref{lem:cross_weight_regeneration} therefore supplies $J_x$ and $J_y$ exactly when a nonzero weight-two or deviatoric input is present. They close with $J_z$ to $\so(3)$, and the rotation orbit of a nonzero traceless symmetric matrix spans the five-dimensional traceless symmetric space, giving $\sll(3,\R)$. Without that input, $\R J_z\oplus\mathcal S_1$ is closed under brackets and is a copy of $\sll(2,\R)$.
\end{proof}

The matrix projection alone does not determine the translation part of an affine algebra. When the projected algebra contains the identity and acts irreducibly, the next lemma relates the translation kernel to the existence of a common fixed point.

\begin{lemma}[Affine-kernel dichotomy]
\label{lem:affine_kernel_dichotomy}
Let $V$ be a finite-dimensional real vector space and let $\mathfrak h\subseteq\gl(V)\ltimes V$ be a Lie algebra with matrix projection $\mathfrak l$. Its pure-translation kernel
\[
T:=\mathfrak h\cap V=\ker\bigl(\mathfrak h\longrightarrow\mathfrak l\bigr)
\]
is an $\mathfrak l$-submodule of $V$. Suppose that $\mathbb I\in\mathfrak l$ and that $V$ is irreducible as an $\mathfrak l$-module. Then exactly one of the following holds:
\begin{enumerate}
    \item $T=V$.
    \item $T=0$, and there is a vector $v\in V$ such that
    \[
    \mathfrak h=\{(A,Av):A\in\mathfrak l\}.
    \]
    In this case all affine vector fields in $\mathfrak h$ have the common fixed point $-v$.
\end{enumerate}
If there is no common fixed point, then $T=V$.
\end{lemma}

\begin{proof}
The kernel $T$ is an ideal of $\mathfrak h$, and $[(A,b),(0,t)]=(0,At)$, so it is an $\mathfrak l$-submodule. Irreducibility gives $T=0$ or $T=V$. If $T=0$, the projection is an isomorphism and $\mathfrak h$ is the graph of a linear map $c:\mathfrak l\to V$. Closure under brackets gives the cocycle identity
\[
c([A,B])=Ac(B)-Bc(A).
\]
Set $v=c(\mathbb I)$. Since $[\mathbb I,A]=0$, the identity with $(\mathbb I,A)$ reads $0=c(A)-Av$, hence $c(A)=Av$. The affine field $x\mapsto Ax+Av$ vanishes at $x=-v$ for every $A$. A nonzero translation has no fixed point, so the two alternatives are exclusive.
\end{proof}

\begin{lemma}[Parametric affine closures for the qubit strata]
\label{lem:qubit_parametric_affine_closure}
Let $G\in\mathcal F$ contain one channel, write $V=\R^3$, and let $\pi_{\mathrm{lin}}:\gl(V)\ltimes V\to\gl(V)$ be the matrix projection. For $\so(2)_{\hat z}$ control, decompose the channel matrix as
\[
Q=\tau\mathbb I+uS_{13}+vS_{23}+pD+\eta S_{12}+rH,
\]
where $S_{ij}=E_{ij}+E_{ji}$, $D=E_{11}-E_{22}$,
$H=\operatorname{diag}(1,1,-2)$, and
$\tau=\operatorname{tr}(Q)/3=-4|a|^2/3\neq0$. Put
$T=\g(G)\cap V$. The matrix projection and translation kernel in the nontrivial channel rows are as follows.
\begin{enumerate}
    \item In $\mathsf{Q4}$ and $\mathsf{Q5}$, the projection is $\sll(2,\R)\oplus\R$. The kernels are $T=0$ and $T=\hat z^\perp$, respectively.
    \item In $\mathsf{Q6}$, the projection is $\sll(2,\R)\oplus\R\mathbb I$ on its irreducible three-dimensional module, and $T=V$.
    \item In $\mathsf{Q7}$, $\mathsf{Q8}$, and $\mathsf{Q12}$, one has $\g=\gl(3,\R)$ and $T=0$.
    \item In $\mathsf{Q9}$, $\mathsf{Q10}$, and $\mathsf{Q13}$, one has $\g=\gl(3,\R)\ltimes V$ and $T=V$.
\end{enumerate}
\end{lemma}

\begin{proof}
We first determine the matrix projections. The operator $\operatorname{ad}(J_z)^2$ has eigenvalues $0$, $-1$, and $-4$ on the weight-zero, weight-one, and weight-two sectors. The corresponding Lagrange polynomials are
\[
\frac{(\operatorname{ad}(J_z)^2+1)(\operatorname{ad}(J_z)^2+4)}{4},\qquad
-\frac{\operatorname{ad}(J_z)^2(\operatorname{ad}(J_z)^2+4)}{3},\qquad
\frac{\operatorname{ad}(J_z)^2(\operatorname{ad}(J_z)^2+1)}{12}
\]
Applying them to $Q$ yields elements of the bracket-generated linear span and extracts the three isotypic components separately. Thus the following sector calculations do not assume that the nonzero weights can be actuated independently. If $m_1=0$ and $m_2\neq0$, then $u=v=0$ and $p^2+\eta^2>0$. With $S_2=pD+\eta S_{12}$ and $C=\tau\mathbb I+rH$, the four matrices
\[
J_z,\qquad S_2,\qquad [J_z,S_2],\qquad C
\]
belong to the projected closure. In normalized coordinates adapted to the antisymmetric, weight-two, and scalar sectors, their determinant is
$6\tau(p^2+\eta^2)\neq0$. The first three close on the standard
$\sll(2,\R)$ acting in the $xy$-plane, while $C$ commutes with them. No bracket leaves these four sectors, so this is the full projection in both $\mathsf{Q4}$ and $\mathsf{Q5}$.

In $\mathsf{Q6}$, the conditions $m_2=\operatorname{dev}=0$ give
$Q=\tau\mathbb I+S_1$ with $S_1=uS_{13}+vS_{23}$ and
$u^2+v^2>0$. The corresponding four-coordinate minor for
\[
J_z,\qquad S_1,\qquad [J_z,S_1],\qquad\tau\mathbb I
\]
is $3\tau(u^2+v^2)\neq0$. The first three matrices form
$\sll(2,\R)$ on its irreducible three-dimensional module. Equivalently, the equations for a matrix commuting with all three have rank eight and leave only scalar matrices. Thus the projected algebra is
$\sll(2,\R)\oplus\R\mathbb I$ throughout the degeneration.

Now assume $m_1\neq0$ and put $Z=pD+\eta S_{12}+rH$. For $Z\neq0$, Lemma~\ref{lem:cross_weight_regeneration} supplies a transverse rotation, including at the nonzero rank-one pairing points. The scalar part drops out of every bracket, so the regeneration argument of Corollary~\ref{cor:axial_closure_dichotomy} gives $\so(3)$ for $Q$ as well. Its adjoint action generates the five-dimensional traceless symmetric sector, hence $\sll(3,\R)$ throughout $\mathsf{Q7}$--$\mathsf{Q10}$. Subtracting the traceless part from $Q$ supplies $\tau\mathbb I$, so the projection is $\gl(3,\R)$. Under full $\so(3)$ control the same conclusion follows directly: rank at most two prevents the nonzero matrix $M$ from being scalar, so the traceless symmetric part of $Q$ is nonzero. Its rotation orbit and the controls span $\sll(3,\R)$, and $\tau\mathbb I$ completes $\gl(3,\R)$. This proves the projected-algebra claim in $\mathsf{Q12}$ and $\mathsf{Q13}$ as well.

It remains to determine $T$. We first record the physical rank--kernel implications used below. Since $M=\alpha\alpha^T+\beta\beta^T$, one always has $\operatorname{rank}M\leq2$. In the non-unital case $\alpha$ and $\beta$ are independent, so $\operatorname{rank}M=2$ and $b=4\alpha\times\beta$ spans $\ker M$. If $m_1=0$, then $M$ is block diagonal on $\hat z^\perp\oplus\R\hat z$. Its one-dimensional kernel then lies wholly in one of these two invariant blocks. Hence $b_\perp\neq0$ forces $b_z=0$. In the unital case $M=\alpha\alpha^T$ has rank one. If $m_2=0$, its planar block $\alpha_\perp\alpha_\perp^T$ is a scalar multiple of the planar identity. The rank bound of at most one forces that scalar to vanish, whence $\alpha_\perp=0$ and $m_1=0$. Thus $m_1\neq0$ and $m_2=0$ is necessarily non-unital, as required in $\mathsf{Q6}$, while every unital branch with $m_1\neq0$, including $\mathsf{Q8}$, has $m_2\neq0$.

A non-unital channel in principal form has the unique Bloch fixed point
\[
v_\ast=\frac{2st}{s^2+t^2}\,\hat\nu,
\qquad Qv_\ast+b=0.
\]
In $\mathsf{Q4}$ the unital branch fixes the origin. In the non-unital branch, $m_1=0$ and $b\parallel\hat z$ force $\hat\nu\parallel\hat z$, so $v_\ast$ is fixed by the control and by every rotated channel. The affine closure therefore has a common fixed point and is the graph over its four-dimensional matrix projection. In particular, $T=0$.

In $\mathsf{Q5}$, the same condition $m_1=0$ makes $Q$ block diagonal relative to $\hat z^\perp\oplus\R\hat z$, while $b_\perp\neq0$ forces $b_z=0$. The half-turn $R_\pi$ about $\hat z$ fixes $Q$ and sends $b$ to $-b$, so two orbit elements satisfy
\[
(Q,b)-\operatorname{Ad}_{R_\pi}(Q,b)=(0,2b).
\]
Bracketing with $(J_z,0)$ gives $(0,2J_zb)$. The planar minor $b_x^2+b_y^2$ is nonzero. Hence $\hat z^\perp\subseteq T$. Conversely, all projected matrices preserve the displayed block decomposition and all channel translations are planar, so no bracket creates a $\hat z$-translation. Thus $T=\hat z^\perp$.

For $\mathsf{Q6}$, $\mathsf{Q9}$, and $\mathsf{Q10}$, non-unitality and $m_1\neq0$ imply $\hat\nu\not\parallel\hat z$. The channel's unique fixed point is therefore not fixed by $J_z$, so the affine family has no common fixed point. Its matrix projection contains $\mathbb I$ and acts irreducibly on $V$: this was proved directly for $\mathsf{Q6}$, and it is immediate for $\gl(3,\R)$ in $\mathsf{Q9}$ and $\mathsf{Q10}$. Lemma~\ref{lem:affine_kernel_dichotomy} gives $T=V$. For $\mathsf{Q13}$ the full rotation algebra has common fixed space $\{0\}$, while a non-unital channel does not fix the origin, so the same lemma again gives $T=V$. The unital rows $\mathsf{Q7}$, $\mathsf{Q8}$, and $\mathsf{Q12}$ are translation-free and have $T=0$.
\end{proof}

These results determine the affine closure. The following corollaries identify its radical and radical wedge, using the unit-mass axial average $w_0$ from \eqref{eq:axial_unit_mass_average}.

\begin{corollary}[Radicals of the qubit strata]
\label{cor:qubit_parametric_radicals}
The radicals in the nontrivial channel rows are
\begin{enumerate}
    \item $\rr=\R w_0$ in $\mathsf{Q4}$ and $\rr=\R w_0\ltimes\hat z^\perp$ in $\mathsf{Q5}$.
    \item $\rr=\R\mathbb I\ltimes V$ in $\mathsf{Q6}$, $\mathsf{Q9}$, $\mathsf{Q10}$, and $\mathsf{Q13}$.
    \item $\rr=\R\mathbb I$ in $\mathsf{Q7}$, $\mathsf{Q8}$, and $\mathsf{Q12}$.
\end{enumerate}
Thus $\dim\rr=1$ in $\mathsf{Q4}$, $\mathsf{Q7}$, $\mathsf{Q8}$, and $\mathsf{Q12}$. One has $\dim\rr=3$ in $\mathsf{Q5}$ and $\dim\rr=4$ in $\mathsf{Q6}$, $\mathsf{Q9}$, $\mathsf{Q10}$, and $\mathsf{Q13}$.
\end{corollary}

\begin{proof}
A graph algebra is isomorphic to its matrix projection, so its radical is the lifted radical of that projection. When $T\neq0$, the inverse image of the projected radical is a solvable ideal and the corresponding quotient is the displayed semisimple algebra. The inverse image is therefore the radical. Lemma~\ref{lem:qubit_parametric_affine_closure} gives the stated formulas and dimensions.
\end{proof}

\begin{corollary}[Radical wedges of the qubit strata]
\label{cor:qubit_parametric_radical_wedges}
The radical wedges in the same rows are as follows.
\begin{enumerate}
    \item In $\mathsf{Q4}$, $\Wrad=\R_{\geq0}w_0$. In $\mathsf{Q5}$, $\Wrad$ is the circular cone \eqref{eq:q5_cone}, which spans $\rr$.
    \item In $\mathsf{Q6}$, $\mathsf{Q8}$, $\mathsf{Q10}$, and $\mathsf{Q12}$, $\Wrad$ is the pointed orbit-average ray. In $\mathsf{Q7}$ and $\mathsf{Q9}$, it is zero.
    \item In $\mathsf{Q13}$, $\Wrad$ is the quadrupole-free average cone of Corollary~\ref{cor:qubit_full_control_radical_wedge}. It spans all four dimensions of $\rr$.
\end{enumerate}
Every listed radical wedge has trivial edge.
\end{corollary}

\begin{proof}
Corollaries~\ref{cor:qubit_axial_radical_wedge} and~\ref{cor:qubit_full_control_radical_wedge} give the $\mathsf{Q5}$, $\mathsf{Q12}$, and $\mathsf{Q13}$ cases. In $\mathsf{Q4}$ the radical is the weight-zero line and the only transverse orbit sector is the nonzero weight-two component. Membership therefore forces its second Fourier moment to vanish and leaves only $\lambda w_0$ with $\lambda\geq0$. Lemma~\ref{lem:magic_locus} gives the zero-versus-ray alternatives in $\mathsf{Q6}$--$\mathsf{Q10}$. Triviality of the edges follows from the trace argument in Lemma~\ref{lem:qubit_orbit_average}.
\end{proof}

\paragraph{Proof of Theorem~\ref{thm:qubit_classification}: exhaustion, realization, and residual moduli.}
The preceding calculations determine the tuple entries on each branch. We now complete Theorem~\ref{thm:qubit_classification} by checking that these branches exhaust the family and that every row is realized, then identify the residual channel moduli.

\begin{proof}
The edge statement is Lemma~\ref{lem:qubit_edge_rigidity}. Lemma~\ref{lem:qubit_normalization} places $\mf{e}_0$ in normal position, and Corollary~\ref{cor:single_jump_principal_form} places the channel in principal phase. Lemma~\ref{lem:qubit_parametric_affine_closure} determines the projected matrix algebra and pure-translation kernel. Corollaries~\ref{cor:qubit_parametric_radicals} and~\ref{cor:qubit_parametric_radical_wedges} determine the radical and radical wedge. These statements hold over the full parameter domain of each nontrivial branch. We organize the remaining discriminant assignment and realization proof by the control algebra, using the normal form \eqref{eq:jump_normal_form}.

\emph{No controls.} A single generator spans a one-dimensional abelian algebra, whatever the channel. The wedge is the ray it generates, the edge is trivial, and $\Wrad = W$. This is $\mathsf{Q1}$. Rows $\mathsf{Q2}$ and $\mathsf{Q11}$ are the jump-free cases. In these cases, $W=\mf e_0$ is a subspace, and the data are immediate.

\emph{Full control $\so(3)$.} Lemma~\ref{lem:qubit_parametric_affine_closure} and Corollaries~\ref{cor:qubit_parametric_radicals}--\ref{cor:qubit_parametric_radical_wedges} give $\mathsf{Q12}$ in the unital branch and $\mathsf{Q13}$ in the non-unital branch, including all affine and radical-wedge entries.

\emph{Control $\so(2)_{\hat z}$ with one channel.} Decompose the pair $(Q, b)$ into weight sectors as in Section~\ref{sec:channel_normal_form}. The linear span of the orbit $\{e^{\theta J_z}\Dcal_L e^{-\theta J_z}\}$ consists of the fixed weight-zero component together with one two-dimensional rotating plane for each nonzero weight-one or weight-two component, with components of equal weight rotating jointly. Brackets act on sectors by
\begin{equation}
\label{eq:sector_brackets}
[\mathcal S_2,\mathcal S_2]\subseteq\R J_z,\qquad
[\mathcal S_1,\mathcal S_1]\subseteq\R J_z,\qquad
[\mathcal S_1,\mathcal S_2\oplus\mathcal S_0]\subseteq\mathcal R_1,\qquad
[\mathcal R_1,\mathcal S_1]\subseteq\mathcal S_2\oplus\mathcal S_0,
\end{equation}
where $\mathcal R_1=\Span\{J_x,J_y\}$ and the affine bracket is given before Definition~\ref{def:W(G)_semigroup}.

If $m_1 = m_2 = 0$ then $M = \mathrm{diag}(c, c, e)$, and the rank bound $\operatorname{rank} M \leq 2$ forces $c = 0$ or $e = 0$, configurations whose translation part lies along $\hat z$. The channel is then axisymmetric about $\hat z$, the two generators commute, and $\g$ is two-dimensional abelian, row $\mathsf{Q3}$. Here $\rr = \g$, the wedge spans it, and the edge line $\R J_z$ lies inside, which gives $\dim\operatorname{span}\Wrad=2$ and $\dim E(\Wrad)=1$.

If $m_1=0\neq m_2$, the same parametric lemma and radical corollaries give $\mathsf{Q4}$ when $b_\perp=0$ and $\mathsf{Q5}$ when $b_\perp\neq0$. In the non-unital branches, $b\parallel\hat\nu$ and the block condition makes $\hat z$ a principal eigenaxis: $b\parallel\hat z$ exactly when $\hat z\parallel\hat\nu$, while $b\perp\hat z$ selects $\hat\alpha$ or $\hat\beta$ (with the planar degeneracy at $q=1$). In the unital branch, $m_1=0\neq m_2$ makes the dephasing axis transverse to $\hat z$.

If $m_1\neq0$, let $Z\in\mathcal S_2\oplus\mathcal S_0$ be the sum of the weight-two and deviatoric projections of $Q$, using the embeddings in Section~\ref{sec:channel_normal_form}. For $Z\neq0$, Lemma~\ref{lem:qubit_parametric_affine_closure} gives the nine-dimensional unital rows $\mathsf{Q7}$ and $\mathsf{Q8}$ or the twelve-dimensional non-unital rows $\mathsf{Q9}$ and $\mathsf{Q10}$. Lemma~\ref{lem:magic_locus} distinguishes the zero and ray radical wedges within each pair. For $Z=0$, the same closure lemma and radical corollaries give the seven-dimensional $\mathsf{Q6}$ row with four-dimensional radical and pointed uniform-average ray.

The thirteen rows exhaust the discriminant combinations. For $m_1=0$, the pair $(m_2,b)$ gives the three cases $\mathsf{Q3}$--$\mathsf{Q5}$. For $m_1\neq0$, $Z=0$ is the necessarily non-unital degeneration $\mathsf{Q6}$. When $Z\neq0$, unitality separates the nine- from the twelve-dimensional closures, and the magic condition separates $\mathsf{Q7}$ from $\mathsf{Q8}$ and $\mathsf{Q9}$ from $\mathsf{Q10}$. Whenever $Z\neq0$ and $\operatorname{dev}=0$, one has $m_2\neq0$. The parametric closure and radical results determine every tuple entry on these branches, so no further tuple can occur. The row witnesses prove realization.

It remains to identify the residual channel-orbit moduli. With no controls, the principal ratio $q=t/s$ is the only channel modulus, giving $\mathsf{Q1}$. Full control removes the principal frame and leaves $q=0$ in $\mathsf{Q12}$ and $q\in(0,1]$ in $\mathsf{Q13}$. The jump-free rows $\mathsf{Q2}$ and $\mathsf{Q11}$ are single control-algebra orbits. For $\mathsf{Q3}$, $m_1=m_2=0$ gives $M=\operatorname{diag}(r,r,u)$, and $\det M=r^2u=0$ leaves, after positive scaling, only axial dephasing and axial extremal damping. For $\mathsf{Q4}$ and $\mathsf{Q5}$ the same eigenvector reduction gives the branches listed in the residual-data display of Section~\ref{sec:qubit_strata}: in $\mathsf{Q4}$, $\hat\nu\parallel\hat z$ for $q>0$, while in $\mathsf{Q5}$ the control axis is $\hat\alpha$ or $\hat\beta$, with the two branches merging when $q=1$.

The exceptional row $\mathsf{Q6}$ is also discrete. After an orthogonal reorientation of the $xy$-plane fixing $\hat z$, its normal form, with $r>0$ and $c\neq0$, is
\[
M=\begin{pmatrix}r&0&c\\0&r&0\\c&0&r\end{pmatrix}.
\]
The conditions $m_2=\operatorname{dev}(M)=0$ and $m_1\neq0$ are then explicit. Positivity and $\operatorname{rank}M\le2$ force $|c|=r$. A reflection removes the sign, and the eigenvalues are $(2r,r,0)$. Thus $q=1/\sqrt2$ and $(x,y,\zeta)=(1/2,0,1/2)$ up to the stated equivalences, so the relative principal frame is fixed. In the unital rows $q=0$, and the dephasing-channel stabilizer leaves only $x=|\hat z\cdot\hat\alpha|^2$: $x=1/3$ is the single orbit $\mathsf{Q8}$, while its complement in $(0,1)$ is $\mathsf{Q7}$. In the remaining non-unital rows, $(q,[\hat z]_{\mathrm{pr}})$ gives the stated continuous moduli. Substitution of the general principal-frame coordinates defined in Section~\ref{sec:qubit_strata} into $\operatorname{dev}(M)=0$ gives the magic equation stated there. Requiring $m_1\neq0$ and $m_2\neq0$ removes the lower-dimensional degenerations, including $\mathsf{Q6}$. This completes the residual-moduli classification.
\end{proof}

\phantomsection
\printbibliography[heading=bibliography]

\end{document}